\documentclass[12pt, a4paper]{article}
\usepackage[utf8]{inputenc}
\usepackage{fontenc}
\usepackage{hyperref}
\usepackage{amsmath}
\usepackage{dcolumn,lscape}
\usepackage{dsfont,amscd,color}
\usepackage{longtable,multicol,tabularx,graphicx,amssymb}
\usepackage{exscale,amsthm,multirow,rotating}
\usepackage{natbib}
\usepackage{bbm}
\usepackage{tikz,dsfont,float}
\usepackage[caption = false]{subfig}
\newtheorem{thm}{Theorem}

\newtheorem{assumption}{Assumption}

\newtheorem{lemma}{Lemma}

 \graphicspath{{figs/}}

\DeclareMathOperator{\plim}{plim}

\DeclareMathOperator{\Diag}{Diag}
\DeclareMathOperator{\argmin}{argmin} 
 
\newcommand{\bs}{\boldsymbol}
\newcommand{\E}{\mathbb{E}}
\newcommand{\tr}{\mathrm{tr}}

\newcommand{\var}{\mathrm{Var}}
\newcommand{\mb}{\mathbb}
\newcommand{\mf}{\mathbf}

\renewcommand{\vec}{\text{vec}}
\DeclareMathOperator{\e}{\varepsilon}

\newcommand{\Y}{\mathbf{Y}}

\begin{document}

\def\spacingset#1{\renewcommand{\baselinestretch}%
	{#1}\small\normalsize} \spacingset{1}


\title{\bf Dynamic Spatiotemporal ARCH Models}
\author{Philipp Otto\\
	\small{Leibniz University Hannover, Germany}\\
	Osman Do\u{g}an\\
	\small{Department of Economics, University of Illinois at Urbana-Champaign, U.S.A.}\\
	S\"uleyman Ta\c{s}p{\i}nar\\
    \small{Department of Economics, Queens College CUNY, New York, U.S.A.}}
\maketitle

\begin{abstract}
	Geo-referenced data are characterized by an inherent spatial dependence due to the geographical proximity. In this paper, we introduce a dynamic spatiotemporal autoregressive conditional heteroscedasticity (ARCH) process to describe the effects of (i) the log-squared time-lagged outcome variable, i.e., the temporal effect, (ii) the spatial lag of the log-squared outcome variable, i.e., the spatial effect, and (iii)  the spatial lag of the log-squared time-lagged outcome variable, i.e., the spatiotemporal effect, on the volatility of an outcome variable. Furthermore, our suggested process allows for the fixed effects over time and space to account for the unobserved heterogeneity. For this dynamic spatiotemporal ARCH model, we derive a generalized method of moments (GMM) estimator based on the linear and quadratic moment conditions of a specific transformation. We show the consistency and asymptotic normality of the GMM estimator, and determine the best set of moment functions. We investigate the finite-sample properties of the proposed GMM estimator in a series of Monte-Carlo simulations with different model specifications and error distributions. Our simulation results show that our suggested GMM estimator has good finite sample properties. In an empirical application, we use monthly log-returns of the average condominium prices of each postcode of Berlin from 1995 to 2015 (190 spatial units, 240 time points) to demonstrate the use of our suggested model. Our estimation results show that  the temporal, spatial and spatiotemporal lags of the log-squared returns have statistically significant effects on the volatility of the log-returns. 
\end{abstract}

\noindent%

\spacingset{1.45} 

\section{Introduction}\label{s1}

In a standard autoregressive conditional heteroskedasticity (ARCH) model, the volatility is modeled as a linear function of the lagged squared outcome variable in order to account for the volatility clustering patterns observed in the outcome variable \citep{Engle:1982, Bollerslev:1992, Engle:1986}. However, when analyzing geo-referenced time series, a further phenomenom occurs -- observations close in space are more similar than observations farther away -- known as Tobler's first law of geography \citep{Tobler70}. From a statistical perspective, this spatial dependence may occur in the means and the volatility of a random process \citep{Otto:2018}. Thus, in this paper, we extend the standard ARCH model to spatiotemporal data by using some tools from spatial econometrics. In our suggested specification, the log-volatility term may depend on (i) the log-squared time-lagged outcome variable, (ii) the higher-order spatial lags of the log-squared outcome variable, (iii)  the higher-order spatial lags of the log-squared time-lagged outcome variable, (iv) exogenous variables, and (v) the unobserved heterogeneity across regions and time. The presence of higher-order spatial lags in our specification indicates that the log-volatility term of a region may depend on the current and time-lagged outcome variables in the neighboring locations in differing ways, depending on the specifications of the associated spatial weights matrices (e.g., different influences from different directions or directional dependence, cf. \citealt{gupta2015inference,merk2021directional}). We refer to this extended model as the dynamic spatiotemporal ARCH model. 

To introduce an estimation approach for our model, we transform the outcome equation so that it is in the form of log-squared terms. We then substitute the log-volatility equation into the the transformed outcome equation to obtain an estimation equation for the log-squared outcome variable. The resulting specification is in the form of a higher-order spatial dynamic panel data model with disturbance terms that may not have a zero mean. We use an orthonormal and a deviation from group-mean operator to wipe out the regional and time fixed effects from the estimation specification. For the estimation of the transformed model that is free of the regional and time fixed effects, we propose a generalized method of moments (GMM) estimator formulated with a set of linear and quadratic moment functions \citep{Lee:2007,lee2010efficient,Lee:2014}. We show that the resulting GMM estimator has the standard large sample properties irrespective of whether the number of time periods is large or finite. When the number of time periods is large, the precision matrix of our GMM estimator simplifies significantly, allowing us to determine a set of linear and quadratic moment functions that can lead to an efficient estimator. We provide such a set of best linear and quadratic moment functions, and establish the asymptotic properties of the resulting best GMM estimator. In a Monte Carlo simulation study, we show that the proposed GMM estimator performs well in finite samples. 

In the literature, \citet{Robinson:2009} introduces the log-square transformation to a cross-sectional spatial stochastic volatility model, and consider a quasi  maximum likelihood (QML) estimation approach for the estimation of the transformed model.\footnote{In the literature, the log-square transformation approach is also used for the estimation of (i) the cross-sectional spatial stochastic volatility models \citep{Taspinar:2021}, and (ii)  the cross-sectional spatial ARCH/GARCH models \citep{Sato:2017,Otto:2019,Otto:2020,Takaki:2021}.} In a similar manner, we may alternatively consider the QML estimation approach rather than the GMM approach for our transformed model \citep{Yu:2008,Yu:2010,Hol:2020}.  Compared to the QML estimation approach, our GMM estimation approach has the following advantages. First, the GMM approach has the computational advantage over the QML approach since the QML estimation involves calculation of the determinant of a Jacobian term at each iteration during the estimation. The computational cost can be especially high when the number of the cross-sectional units is large\footnote{\citet{Lesage:2009} provide some solutions based on various approximation methods to reduce the computation cost significantly.} Second, it is well known that the QML estimator has an asymptotic bias, and thus requires a bias correction approach even when the number of time period is large \citep{Yu:2008,Yu:2010}. Finally, the QML estimator may have poor finite sample properties since the distribution of the log-squared disturbance terms in the transformed model is approximated by a normal distribution. In the time series literature on the volatility models, it has been documented that the QML estimator obtained in this way has poor finite sample properties \citep{JPR:1994, Shephard:1994, Kim:1998, Koopman:1998}. 

In an empirical application, we use a monthly dataset of the real house price returns in Berlin at the zip-code  level over the period, January 1995 to December 2015 to test the  effect of temporal, spatial and spatiotemporal lags of the log-squared returns on the volatility of the log-returns. That is, we analyze the volatility of the house price returns in a real-estate market on a small geographic area of around 900 $km^2$ (see also \citealt{mcmillen2014local,bille2017two,zhang2017quantile}). In Section~\ref{spec}, we show that our dynamic spatiotemporal ARCH model implies a spatial dynamic panel data model for the log-squared returns. Therefore, the presence of spatial, temporal and spatiotemporal effects in the log-squared returns will provide the empirical evidence for our suggested specification. To motivate the presence of these effects on the log-squared returns, Figure \ref{fig:acfs} displays the average log-squared returns over Berlin's zip-codes (the left figure), the estimated temporal autocorrelation of the log-squared returns as a series of boxplots (the center figure), and the estimated spatiotemporal autocorrelation in terms of Moran's $I$ across the time horizon (the right figure). The first figure shows a clustering pattern in the log-squared returns, indicating the presence of spatial dependence. From the ACF estimates, we can observe a clear temporal volatility clustering, while the spatiotemporal dependence is of minor degree, irregularly fluctuating around zero.

By using a first-order version of our dynamic spatiotemporal ARCH process for the local house price returns, we separately identify temporal, spatial and spatiotemporal interaction effects in the log-squared returns. Our estimation results show that the temporal, spatial and spatiotemporal lags of the log-squared returns have statistically significant effects on the log-volatility, and for that there is significant variation in the log-volatility of the real house price returns in Berlin over its zip-codes. This finding is not surprising, because it has been documented in the literature that the spatial dependence in house price variations might arise due to several factors such as migration, equity transfer, spatial arbitrage and spatial patterns in the determinants of house prices \citep{Meen:1999}. These patterns can change with the local infrastructure \citep{chang2021inter}. Recently, \cite{holmes2017pair} and \citet{bashar2021intra} particularly focus on intra-city house prices and show significant temporal and spatial dependence in the growth rates. In contrast to these studies, we focus on the analysis of the log-volatility as a measure of the market risk.

\begin{figure}
	\begin{center}
	    \includegraphics[width = 0.32\textwidth]{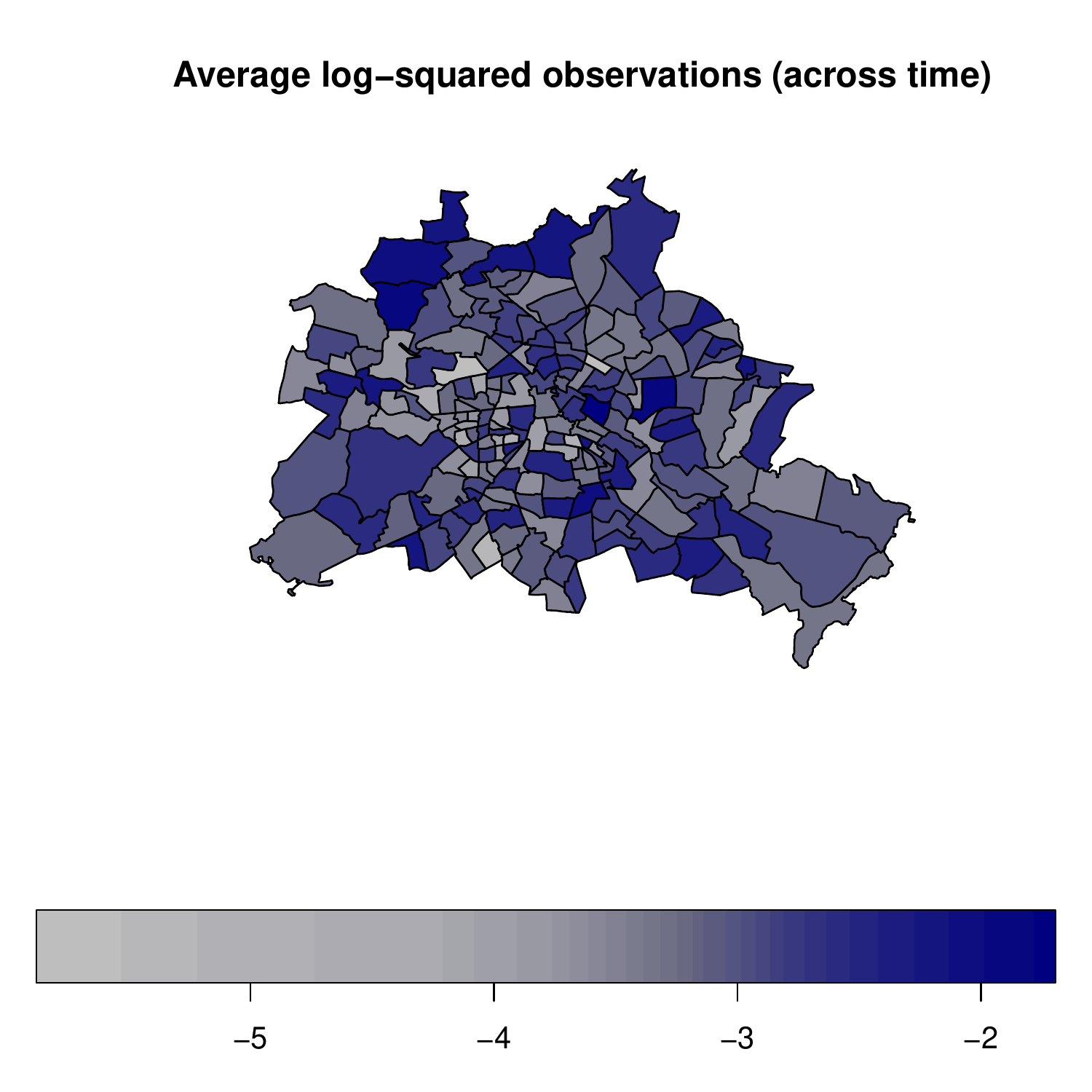}
		\includegraphics[width = 0.32\textwidth]{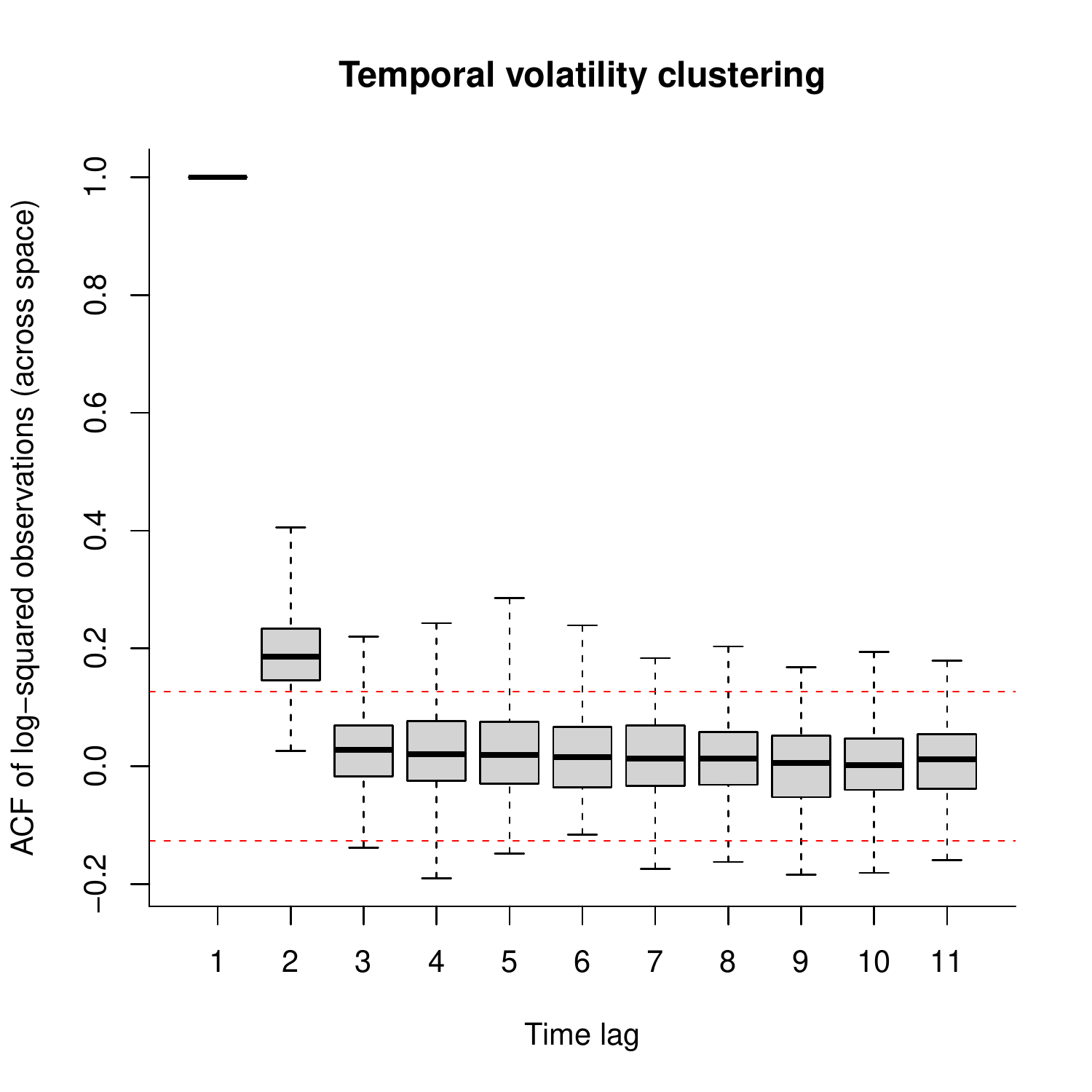}
		\includegraphics[width = 0.32\textwidth]{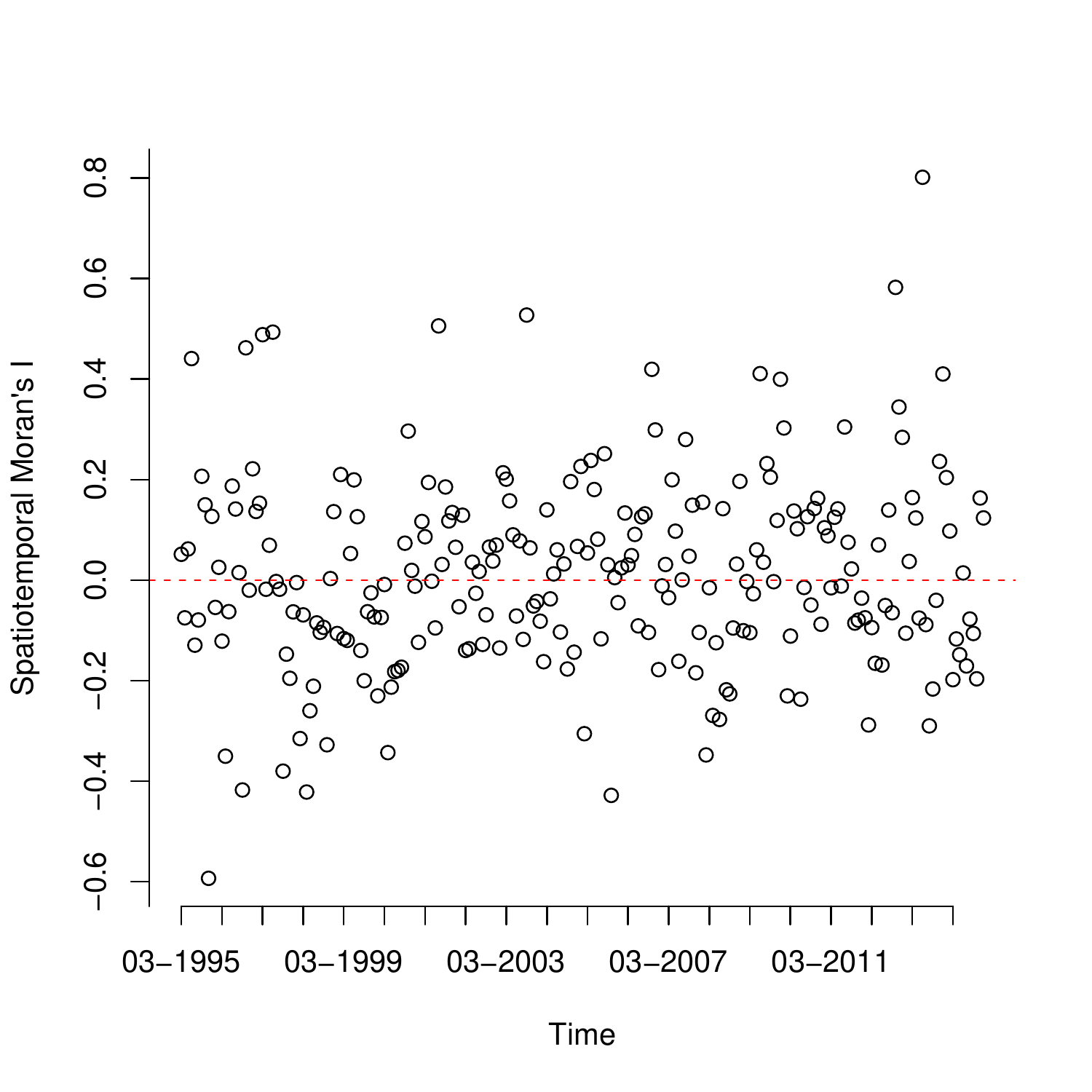}
	\end{center}
	\caption{Indication of spatial, temporal and spatiotemporal volatility clustering. Left: Average log-squared house price return for each of the 190 zip-code areas over the period from February 1995 to December 2015. The map shows a weak clustering effect especially for the outer regions (indicating a positive spatial dependence in the volatility). Center: Temporal ACF depicted as a series of boxplots showing the estimated temporal autocorrelation of all 190 locations. Right: Spatiotemporal correlation in terms of the slope of a regression line between the log squared returns and their temporally lagged neighbors (Moran's $I$ of first spatiotemporal lag). There is no clear pattern with varying coefficients around zero, which may indicate a weak spatiotemporal dependence.}\label{fig:acfs}
\end{figure}

The rest of the paper proceeds in the following way. In Section~\ref{spec}, we state our model specification and discuss its properties. In Section~\ref{gmme}, we provide the details of the GMM estimation approach for our model, and formally establish its large sample properties. In Section~\ref{mc}, we investigate the finite sample properties of our suggested algorithm through an extensive Monte Carlo study. In Section~\ref{emp}, we provide the details of our empirical application on Berlin's house price returns.  In Section~\ref{conc}, we offer our concluding comments with some directions for future studies. Some technical results are relegated to \hyperref[app]{Appendix}.

\section{Model Specification}\label{spec}

The outcome variable $y_{it}$ of  region $i$ at time $t$ is modeled according to
\begin{align}
y_{it}&=h_{it}^{1/2}\e_{it},\label{2.1}\\
\log h_{it}&=\sum_{l=1}^p\sum_{j=1}^n\rho_{l0}m_{l,ij}\log y^2_{jt}+\gamma_0\log y^2_{j,t-1}+\sum_{l=1}^p\sum_{j=1}^n\delta_{l0}m_{l,ij}\log y^2_{j,t-1}\label{2.2}\nonumber\\ 
&\quad+\mathbf{x}^{'}_{it}\bs{\beta}_0+\mu_{i0}+\alpha_{t0},
\end{align}
for $i=1,2,\hdots,n$ and $t=1,\hdots T$. The spatial locations indexed by $i = 1, \ldots, n$ are supposed to be on discrete regular (e.g., for image processes) or irregular lattice, also known as spatial polygons. The latter case is typically present in economics, e.g., regional legal units, districts, countries, etc. Here, $h_{it}$ is considered as the volatility term in region $i$ at time $t$, and $\e_{it}$ are independent and identically distributed random variables that has mean zero and unit variance. The log-volatility terms follow the process in \eqref{2.2},  where $\{m_{l,ij}\}_{l=1}^p$, for $i,j=1,\hdots,n$, are the non-stochastic spatial weights. Here, $p$ is a finite positive integers, and $\{m_{l,ii}\}_{l=1}^p$ are zero for $i=1,\hdots,n$.  The spatial, temporal and spatiotemporal effects of the log-squared outcome variable on the log-volatility are measured by the unknown parameters $\gamma_0$, $\{\rho_{l0}\}_{l=1}^p$, and $\{\delta_{l0}\}_{l=1}^q$, respectively. In \eqref{2.2}, $\mf{x}_i$ is a $k\times1$ vector of exogenous variables with the associated parameter vector $\bs{\beta}_0$, and the regional and time fixed effects are denoted by $\bs{\mu}_0=(\mu_{10},\hdots,\mu_{n0})^{'}$ and $\bs{\alpha}_0=(\alpha_{10},\hdots,\alpha_{T0})^{'}$. Both $\bs{\mu}_0$ and $\bs{\alpha}_0$ can be correlated with the exogenous variables in an arbitrary manner.   We assume that the initial value vector $\mf{Y}_0=(y_{10},\hdots,y_{n0})^{'}$ is observable.

Squaring both sides of \eqref{2.1} and then taking the natural logarithm yield
\begin{align}\label{2.3}
y^{*}_{it}=h^{*}_{it}+\e^{*}_{it},
\end{align}
where $y^*_{it}=\log y^2_{it}$, $h^{*}_{it}=\log h_{it}$ and $\e^*_{it}=\log\e^2_{it}$. In vector form, we can express \eqref{2.3} and \eqref{2.2} as
\begin{align}
&\mf{Y}^{*}_t=\mf{h}^{*}_t+\bs{\e}^{*}_t,\label{2.4}\\
&\mf{h}^{*}_t=\sum_{l=1}^p\rho_{l0}\mf{M}_l\mf{Y}^{*}_t+\gamma_0 \mf{Y}^{*}_{t-1}+\sum_{l=1}^p\delta_{l0}\mf{M}_l\mf{Y}^{*}_{t-1}+\mf{X}_t\bs{\beta}_0+\bs{\mu}_0+\alpha_{t0}\mf{1}_n,\label{2.5}
\end{align}
where $\mf{M}_l=(m_{l,ij})$ is the $n\times n$ spatial weight matrices, $\mf{Y}^{*}_t=(y^{*}_{1t},\hdots,y^{*}_{nt})^{'}$, $\mf{h}^{*}_t=(h^{*}_{1t},\hdots,h^{*}_{nt})^{'}$, $\bs{\e}^{*}_t=(\e_{1t},\hdots,\e^{*}_{nt})^{'}$, $\mf{X}_t=(\mf{x}_{1t},\hdots,\mf{x}_{nt})^{'}$, and $\mf{1}_n$ is the $n\times1$ vector of ones.  The process in \eqref{2.5} indicates that $\mf{h}^{*}_t$ depends on the high order spatial lags of $\mf{Y}^{*}_{t}$ and $\mf{Y}^{*}_{t-1}$. Substituting \eqref{2.5} into \eqref{2.4}, we obtain
\begin{align}\label{2.7}
\mf{Y}^{*}_t=\sum_{l=1}^p\rho_{l0}\mf{M}_l\mf{Y}^{*}_t+\gamma_0 \mf{Y}^{*}_{t-1}+\sum_{l=1}^p\delta_{l0}\mf{M}_l\mf{Y}^{*}_{t-1}+\mf{X}_t\bs{\beta}_0+\bs{\mu}_0+\alpha_{t0}\mf{1}_n+\bs{\e}^{*}_t,
\end{align}
for $t=1,\hdots,T$. This transformed model indicates that our specification in \eqref{2.2} implies a high-order spatial dynamic panel data model for the log-squared outcome variable. In next section, we show how \eqref{2.7} can be used to estimate the parameters of the model.

\section{The Estimation Approach}\label{gmme}
The elements of  $\bs{\e}^{*}_t$ in \eqref{2.7} are i.i.d across $i$ and $t$ but their mean may not be zero. Therefore, we add and subtract $\E\left(\bs{\e}^{*}_t\right)$ to obtain the following equation.
\begin{align}\label{3.1}
\mf{Y}^{*}_t&=\sum_{l=1}^p\rho_{l0}\mf{M}_l\mf{Y}^{*}_t+\gamma_0 \mf{Y}^{*}_{t-1}+\sum_{l=1}^p\delta_{l0}\mf{M}_l\mf{Y}^{*}_{t-1}+\mf{X}_t\bs{\beta}_0+\bs{\mu}_0 +\alpha_{t0}\mf{1}_n+\mu_{\e}\mf{1}_n+\mf{U}_t,
\end{align}
where $\mf{U}_t=(u_{1t},\hdots,u_{nt})^{'}=\bs{\e}^{*}_t-\E\left(\bs{\e}^{*}_t\right)$, and $\mu_{\e}=\E\left(\e^{*}_{it}\right)$. Let $\sigma^2_0=\E(u^2_{it})$. Then, it follows that the elements of $\mf{U}_t$ are i.i.d across $i$ and $t$ with mean zero and variance $\sigma^2_0$. We need to eliminate both fixed effects terms from the model in order to  avoid the incidental parameter problem.  To eliminate $\bs{\mu}$ and $\mu_{\e}\mf{1}_n$ from the model, we consider an orthonormal transformation based on the matrix decomposition of $\mf{J}_T=\left(\mf{I}_T-\frac{1}{T}\mf{1}_T\mf{1}^{'}_T\right)$, where $\mf{I}_T$ is the $T\times T$ identity matrix. Let $\left(\mf{F}_{T,T-1},\frac{1}{\sqrt{T}}\mf{1}_T\right)$ be the orthonormal  eigenvector matrix of $\mf{J}_T$, where $\mf{F}_{T,T-1}$ is the $T\times (T-1)$ sub-matrix  containing eigenvectors corresponding to the eiegenvalues of one.  Let  $\mf{C}=\left(\mf{c}_1,\hdots,\mf{c}_T\right)$ be an $n\times T$ matrix, where $\mf{c}_t$ is an $n\times1$ vector for $t=1,\hdots,T$.  Using $\mf{F}_{T,T-1}$, we can transform $\mf{C}$ into a $n\times(T-1)$ matrix in the following way: $\left(\mf{c}^{*}_1,\hdots,\mf{c}^{*}_{T-1}\right)=\left(\mf{c}_1,\hdots,\mf{c}_T\right)\mf{F}_{T,T-1}$, where $\mf{c}^{*}_j$ is the $j$th column of $\mf{C}\mf{F}_{T,T-1}$ for $j=1,\hdots,T-1$. If we apply  $\mf{F}_{T,T-1}$ to our model in \eqref{3.1} in a similar manner, we obtain
\begin{align}\label{3.2}
\mf{Y}^{**}_t&=\sum_{l=1}^q\rho_{l0}\mf{M}_l\mf{Y}^{**}_t+\gamma _0\mf{Y}^{**,-1}_{t-1}+\sum_{l=1}^q\delta_{l0}\mf{M}_l\mf{Y}^{**,-1}_{t-1}+\mf{X}^{*}_t\bs{\beta}_0+\alpha^{*}_{t0}\mf{1}_n+\mf{U}^{*}_t,
\end{align}
for $t=1,\hdots,T-1$, where  $\left(\mf{Y}^{**}_1,\hdots,\mf{Y}^{**}_{T-1}\right)=\left(\mf{Y}^{*}_1,\hdots,\mf{Y}^{*}_{T}\right)\mf{F}_{T,T-1}$, $\left(\mf{Y}^{**,-1}_0,\hdots,\mf{Y}^{**,-1}_{T-2}\right)=\left(\mf{Y}^{*}_0,\hdots,\mf{Y}^{*}_{T-1}\right)\mf{F}_{T,T-1}$, $\left(\mf{X}^{*}_{l1},\hdots,\mf{X}^{*}_{k,T-1}\right)=\left(\mf{X}_{l1},\hdots,\mf{X}_{lT}\right)\mf{F}_{T,T-1}$, where $\mf{X}_{lt}$ is the $l$th column of $\mf{X}_t$ for $l=1,\hdots,k$, $\left(\alpha^{*}_1,\hdots,\alpha^{*}_{T-1}\right)=\left(\alpha_1,\hdots,\alpha_T\right)\mf{F}_{T,T-1}$ and $\left(\mf{U}^{*}_{1},\hdots,\mf{U}^{*}_{T-1}\right)=\left(\mf{U}_{1},\hdots,\mf{U}_{T}\right)\mf{F}_{T,T-1}$. Note that both $\bs{\mu}_0$ and $\mu_{\e}\mf{1}_n$ are dropped from the model since
\begin{align*}
&\left(\bs{\mu}_0,\hdots,\bs{\mu}_0\right)\mf{F}_{T,T-1}=\bs{\mu}_0\mf{1}^{'}_T\mf{F}_{T,T-1}=\mf{0}_{n\times(T-1)},\\
&\left(\mu_{\e}\mf{1}_n,\hdots,\mu_{\e}\mf{1}_n\right)\mf{F}_{T,T-1}=\mu_{\e}\mf{1}_n\mf{1}^{'}_T\mf{F}_{T,T-1}=\mf{0}_{n\times(T-1)}.
\end{align*}
Let $N=n(T-1)$ and $\mf{U}_N=(\mf{U}^{*'}_1,\hdots,\mf{U}^{*'}_{T-1})^{'}$. Note that we can express  $\left(\mf{U}^{*}_{1},\hdots,\mf{U}^{*}_{T-1}\right)=\left(\mf{U}_{1},\hdots,\mf{U}_{T}\right)\mf{F}_{T,T-1}$ as $\left(\mf{U}^{*'}_{1},\hdots,\mf{U}^{*'}_{T-1}\right)^{'}=\left(\mf{F}^{'}_{T,T-1}\otimes\mf{I}_n\right)\left(\mf{U}^{'}_{1},\hdots,\mf{U}^{'}_{T}\right)^{'}$.\footnote{Note that the matrix equation $\mf{A}\mf{B}\mf{C}=\mf{D}$, where $\mf{D}$, $\mf{A}$, $\mf{B}$, and $\mf{C}$ are suitable matrices, can be expressed as  $\vec(\mf{D})=(\mf{C}^{'}\otimes \mf{A})\vec(\mf{B})$, where $\vec(\mf{B})$ denotes the vectorization of the matrix $\mf{B}$ \citep[p. 282]{Karim:2005}. This property can be applied to $\left(\mf{U}^*_1,\mf{U}^{*}_2,\hdots,\mf{U}^*_{T-1}\right)=\left(\mf{U}_1,\mf{U}_2,\hdots,\mf{U}_T\right)\mf{F}_{T,T-1}$ by setting $\mf{D}=\left(\mf{U}^*_1,\mf{U}^{*}_2,\hdots,\mf{U}^*_{T-1}\right)$, $\mf{C}=\mf{F}_{T,T-1}$, $\mf{B}=\left(\mf{U}_1,\mf{U}_2,\hdots,\mf{U}_T\right)$ and $\mf{A}=\mf{I}_n$. } Then, it follows that 
\begin{align*}
\E\left(\mf{U}_N\mf{U}^{'}_N\right)=\E\left(\left(\mf{F}^{'}_{T,T-1}\otimes\mf{I}_n\right)\left(\mf{U}^{'}_{1},\hdots,\mf{U}^{'}_{T}\right)^{'}\left(\mf{U}^{'}_{1},\hdots,\mf{U}^{'}_{T}\right)\left(\mf{F}_{T,T-1}\otimes\mf{I}_n\right)\right)=\sigma^2_0\mf{I}_N,
\end{align*}
indicating that the elements of $\mf{U}_N$ are uncorrelated. Among the orthonormal transformations, \citet{Lee:2014} show that  the forward orthogonal difference (the Helmert transformation) can be useful for the spatial dynamic panel data models. Thus, we have the following explicit forms for the transformed variables: $\mf{Y}^{**}_t=\left(\frac{T-t}{T-t+1}\right)^{1/2}\left(\mf{Y}^{*}_t-\frac{1}{T-t}\sum_{h=t+1}^T\mf{Y}^{*}_h\right)$, $\mf{Y}^{**,-1}_{t-1}=\left(\frac{T-t}{T-t+1}\right)^{1/2}\left(\mf{Y}^{*}_{t-1}-\frac{1}{T-t}\sum_{h=t}^{T-1}\mf{Y}^{*}_h\right)$ and the other variables are expressed similarly. 

The transformed model in \eqref{3.2} includes the transformed time fixed effects. These terms can be eliminated by pre-multiplying the model with $\mf{J}_n=\left(\mf{I}_n-\frac{1}{n}\mf{1}_n\mf{1}^{'}_n\right)$ to get
\begin{align}\label{3.3}
\mf{J}_n\mf{Y}^{**}_t&=\sum_{l=1}^p\rho_{l0}\mf{J}_n\mf{M}_l\mf{Y}^{**}_t+\gamma _0\mf{J}_n\mf{Y}^{**,-1}_{t-1}+\sum_{l=1}^p\delta_{l0}\mf{J}_n\mf{M}_l\mf{Y}^{**,-1}_{t-1}+\mf{J}_n\mf{X}^{*}_t\bs{\beta}_0+\mf{J}_n\mf{U}^{*}_t,
\end{align}
where we used the fact that $\mf{J}_n\mf{1}_n=\mf{0}_{n}$. Our GMM estimation approach is based on \eqref{3.3}. It is clear that we need to determine IVs for the following terms:  $\{\mf{M}_j\mf{Y}^{**}_t\}_{j=1}^p$, $\mf{Y}^{**,-1}_{t-1}$ and $\{\mf{M}_j\mf{Y}^{**,-1}_{t-1}\}_{j=1}^p$ for $t=1,\hdots,T-1$. That is, we need IVs for the following variables:
\begin{align}\label{3.4}
\mf{J}_n\left(\mb{M}\mf{Y}^{**}_t,\mf{Y}^{**,-1}_{t-1},\mb{M}\mf{Y}^{**,-1}_{t-1}\right),
\end{align}
where $\mb{M}\mf{Y}^{**}_{t}=\left(\mf{M}_1\mf{Y}^{**}_{t},\hdots,\mf{M}_p\mf{Y}^{**}_{t}\right)$  and $\mb{M}\mf{Y}^{**,-1}_{t-1}=\left(\mf{M}_1\mf{Y}^{**,-1}_{t-1},\hdots,\mf{M}_p\mf{Y}^{**,-1}_{t-1}\right)$. Let $\mathcal{F}_{t-1}$ be the $\sigma$-algebra generated by $\left(\mf{Y}_0,\hdots,\mf{Y}_{t-1}\right)$ conditional on $\left(\mf{X}_1,\hdots,\mf{X}_T,\bs{\mu}_0,\bs{\alpha}_0\right)$. Then, we can formulate the theoretical linear IVs based on the expectation of \eqref{3.4} conditional on $\mathcal{F}_{t-1}$.

We use $\bs{\rho}_0=\left(\rho_{10},\hdots,\rho_{p0}\right)^{'}$ and $\bs{\delta}_0=\left(\delta_{10},\hdots,\delta_{p0}\right)^{'}$ to denote the true parameter values, and $\bs{\rho}=\left(\rho_{1},\hdots,\rho_{p}\right)^{'}$ and $\bs{\delta}=\left(\delta_{1},\hdots,\delta_{p}\right)^{'}$ to denote arbitrary parameter values. Let $\mf{S}(\bs{\rho})=\left(\mf{I}_n-\sum_{l=1}^p\rho_{l}\mf{M}_l\right)$, $\mf{S}\equiv\mf{S}(\bs{\rho}_0)$, $\mf{A}(\bs{\rho},\bs{\delta},\gamma)=\mf{S}^{-1}(\bs{\rho})\left(\gamma\mf{I}_n+\sum_{l=1}^p\delta_l\mf{M}_l\right)$, and $\mf{A}\equiv\mf{A}(\bs{\rho}_0,\bs{\delta}_0,\gamma_0)$. Then, the reduced form of \eqref{3.2} can be expressed as 
 \begin{align}
 \mf{Y}^{**}_t&=\mf{A}\mf{Y}^{**,-1}_{t-1}+\mf{S}^{-1}\left(\mf{X}^{*}_t\bs{\beta}_0+\alpha^{*}_{t0}\mf{1}_n+\mf{U}^{*}_t\right).
 \end{align}
 Let $\mf{Z}^{**}_t=\left(\mf{Y}^{**,-1}_{t-1},\mb{M}\mf{Y}^{**,-1}_{t-1},\mf{X}^{*}_t\right)$ be the $n\times k_z$ matrix, where $k_z=p+k+1$, and $\mf{Z}_N=\left(\mf{Z}^{**'}_1,\hdots,\mf{Z}^{**'}_{T-1}\right)^{'}$. Then, using \eqref{3.2}, we have 
 \begin{align}
  &\mf{M}_r\mf{Y}^{**}_t= \mf{G}_{r}\left(\mf{Z}^{**}_t\bs{\eta}_0+\alpha^{*}_{t0}\mf{1}_n\right)+\mf{G}_{r}\mf{U}^{*}_t,\quad r=1,\hdots,p,\label{3.6}
 \end{align}
 where $\bs{\eta}_0=(\gamma_0,\bs{\delta}^{'}_0,\bs{\beta}^{'}_0)^{'}$, and $\mf{G}_{r}=\mf{M}_r\mf{S}^{-1}$. We can use \eqref{3.6} to determine IVs for  $\{\mf{M}_j\mf{Y}^{**}_t\}_{j=1}^p$. In the case of $\mf{Y}^{**,-1}_{t-1}$, we can use all strictly exogenous variables $\mf{X}^{*}_s$ for $s=1,\hdots,T-1$, and the time lag variables $\mf{Y}^{*}_{0},\hdots,\mf{Y}^{*}_{t-1}$ as IVs.  Similarly, we can use $\mf{M}_j\mf{X}^{*}_s$ for $s=1,\hdots,T-1$, and $\mf{M}_j\mf{Y}^{*}_s$ for $s=0,1,\hdots,t-1$ as IVs for $\mf{M}_j\mf{Y}^{**,-1}_{t-1}$. Let $\mf{Q}_t$ be the $n\times k_q$ matrix of IVs for $t=1,\hdots,T-1$, where $k_q\geq k+2p+1$. For example, we may choose $\mf{Q}_t$ as 
\begin{align}
\left(\Y^{*}_{t-1},\,\mb{M}\Y^{*}_{t-1},\mb{M}^2\Y^{*}_{t-1},\mf{X}^{*}_t,\mb{M}\mf{X}^{*}_t,\mb{M}^2\mf{X}^{*}_t\right),
\end{align}
where $\mb{M}^2\Y^{*}_{t-1}=\left(\mf{M}^2_1\Y^{*}_{t-1},\hdots,\mf{M}_1\mf{M}_p\Y^{*}_{t-1},\mf{M}_2\mf{M}_1\Y^{*}_{t-1},\hdots,\mf{M}_1\mf{M}_p\Y^{*}_{t-1},\hdots,\mf{M}^2_p\Y^{*}_{t-1}\right)$, and $\mb{M}^2\mf{X}^{*}_t$  is defined similarly.  Denote $\mf{Q}_N=(\mf{Q}^{'}_1,\hdots,\mf{Q}^{'}_{T-1})^{'}$, $\mf{J}_N=\mf{I}_{T-1}\otimes\mf{J}_n$ and $\mf{S}_N(\bs{\rho})=\mf{I}_{T-1}\otimes\mf{S}(\bs{\rho})$. Then, the linear moment conditions based on $\mf{Q}_N$ can be formulated as 
\begin{align}
\mf{Q}^{'}_N\mf{J}_N\mf{U}_N(\bs{\theta}),
\end{align}
where $\bs{\theta}=\left(\bs{\rho}^{'},\bs{\eta}^{'}\right)^{'}$, $\mf{U}_N(\bs{\theta})=\left(\mf{U}^{*'}_1(\bs{\theta}),\hdots,\mf{U}^{*'}_{T-1}(\bs{\theta})\right)^{'}$ and $\mf{U}^{*}_t(\bs{\theta})=\mf{S}(\bs{\rho})\mf{Y}^{**}_t-\mf{Z}^{**}_t\bs{\eta}-\alpha^{*}_{t}\mf{1}_n$. Note that the transformed time fixed effects $\bs{\alpha}^{*}=(\alpha^{*}_{1},\hdots,\alpha^{*}_{T-1})^{'}$ will be eliminated in the moment function because $\mf{U}_N(\bs{\theta})$ is pre-multiplied  by $\mf{J}_N$.

Following \citet{Lee:2007} and \citet{Lee:2014}, we also consider the quadratic moment functions for estimation. The quadratic moment functions are based on the idea that the vector $\mf{P}_{l}\mf{J}_n\mf{U}^{*}_t$ can be uncorrelated with  $\mf{J}_n\mf{U}^{*}_t$ for an $n\times n$ matrix $\mf{P}_l$ satisfying $\tr\left(\mf{J}_n\mf{P}_l\mf{J}_n\right)=0$, while it may be correlated with  $\mf{G}_{r}\mf{U}^{*}_t$ in \eqref{3.6}. Let $\mf{P}_{lN}=\mf{I}_{T-1}\otimes \mf{P}_l$, and assume that there are $m$ such quadratic moment matrices. Then, the quadratic moment functions can be expressed as 
\begin{align}\label{3.10}
\mf{U}^{'}_N(\bs{\theta})\mf{J}_N\mf{P}_{lN}\mf{J}_N\mf{U}_N(\bs{\theta}),
\end{align}
for $l=1,2,\hdots,m$. Combining the linear and quadratic moment functions, we obtain the following vector of moment functions,
\begin{align}
\mf{g}_N(\bs{\theta})=
\begin{pmatrix}
\mf{U}^{'}_N(\bs{\theta})\mf{J}_N\mf{P}_{1N}\mf{J}_N\mf{U}_N(\bs{\theta})\\
\vdots\\
\mf{U}^{'}_N(\bs{\theta})\mf{J}_N\mf{P}_{mN}\mf{J}_N\mf{U}_N(\bs{\theta})\\
\mf{Q}^{'}_N\mf{J}_N\mf{U}_N(\bs{\theta})
\end{pmatrix}.
\end{align}
Let $\vec(\mf{P})$ be the vectorization of the square matrix $\mf{P}$, $\vec_D(\mf{P})$ be the column vector formed from the diagonal elements of $\mf{P}$ and $\mf{P}^s=\mf{P}+\mf{P}^{'}$. Define $\bs{\Omega}_N=\frac{1}{N}\E\left(\mf{g}_N(\bs{\theta}_0)\mf{g}_N(\bs{\theta}_0)\right)$. Then, using Lemma~\ref{l1}, it can be shown that\footnote{In applying Lemma~\ref{l1}, we use the fact that $\tr\left(\mf{A}^{'}\mf{B}\right)=\vec^{'}\left(\mf{A}\right)\vec\left(\mf{B}\right)=\vec^{'}\left(\mf{B}\right)\vec\left(\mf{A}\right)$, where $\mf{A}$ and $\mf{B}$ are any two $N\times N$ matrices.} 
\begin{align}
\bs{\Omega}_N=\plim_{n\to\infty}\frac{\sigma^4_0}{N}
\begin{pmatrix}
\bs{\Delta}_{mN}&\mf{0}_{m\times q}\\
\mf{0}_{q\times m}&\frac{1}{\sigma^2_0}\mf{Q}^{'}_N\mf{J}_N\mf{Q}_N
\end{pmatrix}
+\lim_{n\to\infty}\frac{\mu_4-3\sigma^4_0}{N}
\begin{pmatrix}
\bs{\omega}^{'}_{mN}\bs{\omega}_{mN}&\mf{0}_{m\times q}\\
\mf{0}_{q\times m}&\mf{0}_{q\times q}
\end{pmatrix},
\end{align}
where $\mu_4$ is the fourth moment of $u_{it}$, $\bs{\omega}_{mN}=\left(\vec_D\left(\mf{J}_N\mf{P}^{'}_{1N}\mf{J}_N\right),\hdots,\vec_D\left(\mf{J}_N\mf{P}^{'}_{mN}\mf{J}_N\right)\right)$ and
\begin{align*}
\bs{\Delta}_{mN}&=\left(\vec\left(\mf{J}_N\mf{P}^{'}_{1N}\right),\hdots,\vec\left(\mf{J}_N\mf{P}^{'}_{mN}\mf{J}_N\right)\right)^{'}\\
&\times \left(\vec\left(\mf{J}_N\mf{P}^s_{1N}\mf{J}_N\right),\hdots,\vec\left(\mf{J}_N\mf{P}^s_{mN}\mf{J}_N\right)\right).
\end{align*}
 Let $\hat{\bs{\Omega}}_N$ be a consistent estimator of $\bs{\Omega}_N$, i.e., $\hat{\bs{\Omega}}_N-\bs{\Omega}_N=o_p(1)$. Then, the optimal GMM estimator is defined as 
\begin{align}\label{3.12}
\hat{\bs{\theta}}_N=\argmin_{\bs{\theta}\in\bs{\Theta}}\mf{g}^{'}_N(\bs{\theta})\hat{\bs{\Omega}}^{-1}_N\mf{g}_N(\bs{\theta}).
\end{align}
To investigate the asymptotic properties of $\hat{\bs{\theta}}_N$, we require the following assumptions.
\begin{assumption}\label{a1}
The disturbance terms $u_{it}$, for $i=1,2,\hdots,n$, and $t=1,2,\hdots,T$, are i.i.d. across $i$ and $t$ with mean zero, variance $\sigma^2_0$ and $\E\left(|u_{it}|^{4+\kappa}\right)$ for some $\kappa>0$.
\end{assumption}
\begin{assumption}\label{a2}
The spatial weights matrices $\{\mf{M}_l\}_{l=1}^q$ are uniformly bounded in both row and column sums in absolute value.
\end{assumption}
\begin{assumption}\label{a3}
(i) $\mf{S}(\bs{\rho})$ is invertible for all $\bs{\rho}\in\bs{\Lambda}$, where $\bs{\Lambda}$ is a compact parameter space, and $\bs{\rho}_0$ is in the interior of $\bs{\Lambda}$. (ii) $\mf{S}^{-1}(\bs{\rho})$ is uniformly bounded in both row and column sums in absolute value.
\end{assumption}
\begin{assumption}\label{a4}
(i) $\mf{X}_t$ is non-stochastic with $\sum_{t=1}^T\sum_{i=1}^n|x_{it,l}|^{2+\epsilon}<\infty$ for some $\epsilon>0$, where $x_{it,l}$ is the $(i,t)$th element of the $l$th regressor for $l=1,\hdots,k$. Moreover, \linebreak $\lim_{n\to\infty}\frac{1}{N}\mf{X}^{'}_N\mf{J}_N\mf{X}_N$ exists and is non-singular, where $\mf{X}_N=\left(\mf{X}^{*'}_1,\hdots,\mf{X}^{*'}_{T-1}\right)^{'}$. (ii) $\bs{\mu}_0$ and $\bs{\alpha}_0$ are non-stochastic with $\sup_n\frac{1}{n}\sum_{i=1}^n|\mu_{i0}|^{2+\epsilon}<\infty$ and $\sup_T\frac{1}{T}\sum_{t=1}^T|\alpha_{t0}|^{2+\epsilon}<\infty$. 
\end{assumption}
\begin{assumption}\label{a5}
(i) $\Y^{*}_0=\sum_{h=0}^{\bar{h}}\mf{A}^h\mf{S}^{-1}\left(\mf{X}_{-h}\bs{\beta}_0+\bs{\mu}_0+\alpha_{-h0}\mf{1}_n+\mu_{\e}\mf{1}_n+\mf{U}_{-h}\right)$, where $\bar{h}$ can be finite or infinite. (ii) $\sum_{h=0}^{\infty}\text{abs}(\mf{A}^h)$ is uniformly bounded in both row and column sums in absolute value, where the $(i,j)$th element of $\text{abs}(\mf{A})$ is given by $|A_{ij}|$ and $A_{ij}$ is the $(i,j)$th element of $\mf{A}$.
\end{assumption}
\begin{assumption}\label{a6}
$\E\left( \mf{Q}_t|\mathcal{F}_{t-1}\right)=\mf{Q}_t$ and $\E\left(|q_{it,l}|^{2+\epsilon}\right)<\infty$, where $q_{it,l}$ is the $(i,t)$th element of the $l$th column of $\mf{Q}_t$. Moreover, $\plim_{n\to\infty}\frac{1}{N}\mf{Q}^{'}_N\mf{J}_N\left(\mf{Z}_N,\,\mf{L}_{N}\right)$ and \linebreak $\plim_{n\to\infty}\frac{1}{N}\mf{Q}^{'}_N\mf{J}_N\mf{Q}_N$ have full column ranks, where  $\mf{L}_{N}=\left(\mf{L}^{'}_{1},\hdots,\mf{L}^{'}_{T-1}\right)^{'}$ with $\mf{L}_{t}=\left(\mf{L}_{1,t},\hdots,\mf{L}_{p,t}\right)$ and $\mf{L}_{r,t}= \mf{G}_{r}\left(\mf{Z}^{**}_t\bs{\eta}_0+\alpha^{*}_t\mf{1}_n\right)$ for $r=1,2,\hdots,p$.
\end{assumption}
Assumption~\ref{a1} specifies the distribution of the elements of $\mf{U}_t$ for $t=1,2,\hdots,T$. The moment condition in this assumption is required for showing the asymptotic distribution of our set of moment functions. Assumptions~\ref{a2} and ~\ref{a3} are standard assumptions adopted in the literature for limiting the degree of spatial correlation at a manageable degree, e.g., among others, see \citet{KP:2010, Lee:2004}.    Assumption~\ref{a4} provides the regularity conditions for $\mf{X}_N$, $\bs{\mu}_0$ and $\bs{\alpha}_0$. The first part of Assumption~\ref{a5} specifies $\Y^{*}_0$, and the remaining parts are required to limit dependence over time and across cross section units (see \citet{Lee:2014} for the details). The sufficient conditions for the first part of Assumption~\ref{a3}, and the second part of \ref{a5} can be determined. Let $\Vert\cdot\Vert$ be any matrix norm. Then, the following respective conditions will be sufficient for ensuring these parts: (i) $\Vert\sum_{l=1}^p\rho_{l}\mf{M}_l\Vert<1$ and (ii) $\Vert\mf{A}(\bs{\rho},\bs{\delta},\gamma)\Vert<1$.  Note that 
\begin{align*}
\left\Vert\sum_{j=1}^{p}\rho_{l}\mf{M}_l\right\Vert\leq|\rho_1|\cdot\Vert\mf{M}_{1}\Vert+\hdots+|\rho_{p}|\cdot\Vert \mf{M}_{p}\Vert\leq\left(\sum_{j=1}^{p}|\rho_{j}|\right)\times\max_{1\leq j\leq p}\Vert\mf{M}_{_j}\Vert,
\end{align*}
Thus, a relatively restrictive condition for (i) is $\left(\sum_{j=1}^{p}|\rho_{j}|\right)\times\max_{1\leq j\leq p}\Vert\mf{M}_{_j}\Vert<1$. Similarly, we have 
\begin{align*}
&\left\Vert\mf{A}(\bs{\rho},\bs{\delta},\gamma)\right\Vert\leq\left\Vert \mf{S}^{-1}(\bs{\rho})\right\Vert\times\left\Vert\gamma \mf{I}_n+\sum_{l=1}^{p}\delta_{l}\mf{M}_{l}\right\Vert\\
&=\left\Vert \mf{I}_n+\left(\sum_{l=1}^{p}\rho_{l} \mf{M}_{l}\right)+\left(\sum_{l=1}^{p}\rho_{l} \mf{M}_{l}\right)^2+\hdots \right\Vert\times\left\Vert\gamma \mf{I}_n+\sum_{l=1}^{p}\delta_{l}\mf{M}_{l}\right\Vert\\
&\leq\frac{1}{1-\tau_1}\times\left(|\gamma|+ \left(\sum_{l=1}^{p}|\delta_{l}|\right)\times\max_{1\leq l\leq p}\Vert\mf{M}_{l}\Vert\right),
\end{align*} 
where $\tau_1=\left(\sum_{l=1}^{p}|\rho_{l}|\right)\times\max_{1\leq l\leq p}\left\Vert \mf{M}_{l}\right\Vert<1$ is guaranteed by the first condition. This result suggests that a relatively restrictive condition for (ii) is $\frac{1}{1-\tau_1}\times\left(|\gamma|+ \left(\sum_{l=1}^{p}|\delta_{l}|\right)\times\max_{1\leq l\leq p}\Vert\mf{M}_{l}\Vert\right)<1$. When the spatial weights matrices are row normalized these relatively restrictive conditions can be further simplified. For example, if we use the matrix row sum norm, we will get the following sufficient conditions: (i) $\left(\sum_{j=1}^{p}|\rho_{j}|\right)<1$ and (ii) $\left(\sum_{j=1}^{p}|\rho_{j}|+|\gamma|+ \sum_{l=1}^{p}|\delta_{l}|\right)<1$.

Assumption~\ref{a6} provides the regularity conditions for the IV matrix $\mf{Q}_t$. The first part states that $\mf{Q}_t$ is pre-determined in the sense that $\E\left( \mf{Q}_t|\mathcal{F}_{t-1}\right)=\mf{Q}_t$. The moment condition in this assumption is required for the application of a CLT to the set of our moment functions (see the CLT given in Lemma~\ref{l3}). The full column rank condition in Assumption~\ref{a6} gives the identification condition based on the linear moment function in our setting. See Appendix~\ref{A.C} for the details on the identification condition in our setting.

Let $\frac{\partial \mf{g}_N(\bs{\theta})}{\partial\bs{\theta}^{'}}=\left(\frac{\partial \mf{g}_N(\bs{\theta})}{\partial\bs{\rho}^{'}},\frac{\partial \mf{g}_N(\bs{\theta})}{\partial\bs{\eta}^{'}}\right)$. In Section~\ref{pt1} of Appendix,  we show that $\frac{1}{N}\frac{\partial \mf{g}_N(\bs{\theta}_0)}{\partial\bs{\theta}^{'}}=\mf{D}_{1N}+\mf{D}_{2N}+O_p(N^{-1/2})$, where $\mf{D}_{1N}=O(1)$ and $\mf{D}_{2N}=O(T^{-1})$.\footnote{The explicit forms of $\mf{D}_{1N}$ and $\mf{D}_{2N}$ are given in Section~\ref{pt1} of Appendix.} The following result gives the limiting distribution of  $\hat{\bs{\theta}}_N$ under the large $T$ and  finite $T$ cases. 

\begin{thm}\label{t1}
Under Assumptions~\ref{a1}-\ref{a6}, we have the following results,
\begin{enumerate}
\item When $T$ is finite and $n\to\infty$, we have
\begin{align}
\sqrt{n}\left(\hat{\bs{\theta}}_N-\bs{\theta}_0\right)\xrightarrow{d}N\left(\mf{0}_{k_z+p},\,\plim_{n\to\infty}\frac{1}{T-1}\left(\left(\mf{D}_{1N}+\mf{D}_{2N}\right)^{'}_N\bs{\Omega}^{-1}_N\left(\mf{D}_{1N}+\mf{D}_{2N}\right)\right)^{-1}\right).
\end{align}
\item When $T\to\infty$ and $n\to\infty$, we have
\begin{align}\label{3.14}
\sqrt{N}\left(\hat{\bs{\theta}}_N-\bs{\theta}_0\right)\xrightarrow{d}N\left(\mf{0}_{k_z+p},\,\plim_{n,T\to\infty}\left(\mf{D}^{'}_{1N}\bs{\Omega}^{-1}_N\mf{D}_{1N}\right)^{-1}\right).
\end{align}
\end{enumerate}
\end{thm}
\begin{proof}
See  Section~\ref{pt1} of Appendix.
\end{proof}
Our estimator defined in \eqref{3.12} requires a consistent estimator of $\Omega_N$. We can use a plug-in estimator of $\Omega_N$ based on an initial GMM estimator, or alternatively, we can formulate a 2SLS estimator based on $\mf{Q}_N$. Let $\mf{Y}_N=\left(\mf{Y}^{**'}_1,\hdots,\mf{Y}^{**'}_{T-1}\right)^{'}$ and  $\mf{Z}_N=\left(\mf{Z}^{**'}_1,\hdots,\mf{Z}^{**'}_{T-1}\right)^{'}$. Then, the 2SLS estimator is
\begin{align}
\tilde{\bs{\theta}}_N=\left(\left(\mb{M}_N\mf{Y}_{N},\mf{Z}_N\right)^{'}\mf{M}_{\mf{Q}}\left(\mb{M}_N\mf{Y}_{N},\mf{Z}_N\right)\right)^{-1}\left(\mb{M}_N\mf{Y}_{N},\mf{Z}_N\right)^{'}\mf{M}_{\mf{Q}}\mf{Y}_N,
\end{align}
where $\mf{M}_{\mf{Q}}=\mf{J}_N\mf{Q}_N\left(\mf{Q}^{'}_N\mf{J}_N\mf{Q}_N\right)^{-1}\mf{Q}^{'}_N\mf{J}_N$ and 
$\mb{M}_N\mf{Y}_{N}=\left(\mf{M}_{1N}\mf{Y}_N,\hdots,\mf{M}_{pN}\mf{Y}_N\right)$ with $\mf{M}_{jN}=\mf{I}_{T-1}\otimes\mf{M}_j$ for $j=1,2,\hdots,p$. Our Theorem~\ref{t1} suggests that 
\begin{align}
\sqrt{n}\left(\tilde{\bs{\theta}}_N-\bs{\theta}_0\right)\xrightarrow{d}N\left(\mf{0},\,\sigma^2_0\plim_{n\to\infty}\left(\frac{1}{N}\left(\mf{L}_N,\mf{Z}_N\right)^{'}\mf{M}_{\mf{Q}}\left(\mf{L}_N,\mf{Z}_N\right)\right)^{-1}\right).
\end{align}
In the case of both the initial GMM and 2SLS estimators, the linear IV matrix can be $\mf{Q}_t=\left(\mf{Y}^{*}_{t-1},\mb{M}\mf{Y}^{*}_{t-1},\mb{M}^2\mf{Y}^{*}_{t-1}, \mf{X}^{*}_t,\mb{M} \mf{X}^{*}_t,\mb{M}^2 \mf{X}^{*}_t\right)$ for $t=1,2,\hdots,T-1$. We may consider the following quadratic moment matrices for the initial GMM estimator: $\mf{P}_j=\left(\mf{M}_j-\frac{\tr(\mf{M}_j\mf{J}_n)}{n-1}\mf{J}_n\right)$ and  $\mf{P}_{j+p}=\left(\mf{M}^2_j-\frac{\tr(\mf{M}^2_j\mf{J}_n)}{n-1}\mf{J}_n\right)$ for $j=1,2,\hdots,p$. Then,  the initial GMM estimator is given by $\tilde{\bs{\theta}}_N=\argmin_{\bs{\theta}\in\bs{\Theta}}\mf{g}^{'}_N(\bs{\theta})\mf{g}_N(\bs{\theta})$ where 
\begin{align}
\mf{g}_N(\bs{\theta})=
\begin{pmatrix}
\mf{U}^{'}_N(\bs{\theta})\mf{J}_N\mf{P}_{1N}\mf{J}_N\mf{U}_N(\bs{\theta})\\
\vdots\\
\mf{U}^{'}_N(\bs{\theta})\mf{J}_N\mf{P}_{2pN}\mf{J}_N\mf{U}_N(\bs{\theta})\\
\mf{Q}^{'}_N\mf{J}_N\mf{U}_N(\bs{\theta})
\end{pmatrix}.
\end{align}

We can use $\tilde{\bs{\theta}}_N$ to formulate the plug-in estimator of  $\Omega_N$, which requires the estimators of $\sigma^2_0$ and $\mu_4$. Let $\tilde{\mf{V}}_t=\mf{S}(\tilde{\bs{\rho}}_N)\mf{Y}^{**}_t-\mf{Z}^{**}_t\tilde{\bs{\eta}}_N$. Then, we can estimate $\sigma^2_0$ by $\tilde{\sigma}^2_N=\frac{1}{N}\sum_{t=1}^{T-1}\tilde{\mf{V}}^{'}_t\mf{J}_n\tilde{\mf{V}}_t$. Let $\Delta\tilde{\mf{V}}_t=\mf{S}(\tilde{\bs{\rho}}_N)\Delta\mf{Y}^{**}_t-\Delta\mf{Z}^{**}_t\tilde{\bs{\eta}}_N$. Then, following \citet{Lee:2014}, we can estimate $\mu_4$ by $\tilde{\mu}_4=\frac{1}{2N}\sum_{i=1}^n\sum_{t=2}^{T}\left(\left[\mf{J}_n\Delta\tilde{\mf{V}}_t\right]_i\right)^4-3\tilde{\sigma}^4$, where $\left[\mf{J}_n\Delta\tilde{\mf{V}}_t\right]_i$ is the $i$th element of $\mf{J}_n\Delta\tilde{\mf{V}}_t$.

Our set of moment functions in \eqref{3.10} depends on the IV matrix $\mf{Q}_N$ and the quadratic moment matrices $\mf{P}_{lN}$ for $l=1,2,\hdots,m$. The asymptotic efficiency of $\hat{\bs{\theta}}_N$ should be considered in choosing the IV and quadratic moment matrices. The best set of IV and quadratic moment matrices is the set that leads to the most efficient GMM estimator. When $T$ is large, the precision matrix $\hat{\bs{\theta}}_N$ takes a simple form allowing for determining the best set of IV and quadratic moment matrices.  When $T$ is large, the proof of Theorem~\ref{t1} indicates that $\frac{1}{N}\frac{\partial \mf{g}_N(\bs{\theta}_0)}{\partial\bs{\theta}^{'}}=\mf{D}_{1N}+O_p(N^{-1/2})$, where
\begin{align}
\mf{D}_{1N}=-\frac{1}{N}
\begin{pmatrix}
\sigma^2_0\mf{C}_{N}&\mf{0}_{m\times k_z}\\
\mf{Q}^{'}_N\mf{J}_N\mf{L}_{N}&\mf{Q}^{'}_N\mf{J}_N\mf{Z}_{N}
\end{pmatrix}.
\end{align}
Then, \eqref{3.14} shows that the precision matrix of $\sqrt{N}\left(\hat{\bs{\theta}}_N-\bs{\theta}_0\right)$ is 
\begin{align}
\mf{D}^{'}_{1N}\bs{\Omega}^{-1}_N\mf{D}_{1N}&=\frac{1}{N}
\begin{pmatrix}
\mf{C}^{'}_N\left(\bs{\Delta}_{mN}+\frac{\mu_4-3\sigma^4_0}{\sigma^4_0}\bs{\omega}^{'}_{mN}\bs{\omega}_{mN}\right)^{-1}\mf{C}_N&\mf{0}_{p\times k_z}\\
\mf{0}_{k_z\times p}&\mf{0}_{k_z\times k_z}
\end{pmatrix}\nonumber\\
&+\frac{1}{N\sigma^2_0}\left(\mf{L}_N,\, \mf{Z}_N\right)^{'}\mf{M}_{\mf{Q}}\left(\mf{L}_N,\, \mf{Z}_N\right).
\end{align}
Since the above precision matrix has the same form as the one given in \citet{Lee:2014}, we use their  approach to determine the best set of quadratic moment matrices.  We will choose the best quadratic matrices by maximizing $\mf{C}^{'}_N\left(\bs{\Delta}_{mN}+\frac{\mu_4-3\sigma^4_0}{\sigma^4_0}\bs{\omega}^{'}_{mN}\bs{\omega}_{mN}\right)^{-1}\mf{C}_N$. As shown in \citet{Lee:2014}, these matrices are 
\begin{align}
\mf{P}^{*}_j=\left(\mf{G}_j-\frac{\tr\left(\mf{G}_j\mf{J}_n\right)}{n-1}\mf{J}_n\right)+c\left(\Diag\left(\mf{J}_n\mf{G}_j\mf{J}_n\right)-\frac{\tr\left(\mf{G}_j\mf{J}_n\right)}{n}\mf{I}_n\right),\quad j=1,2,\hdots,p,
\end{align}
where $c=\left(\frac{n}{n-2}\right)^2\left(\frac{1}{n/(n-2)+(\eta_4-3)/2}-\frac{n-2}{n}\right)$ and  $\eta_4=\mu_4/\sigma^4_0$. 

In the case of the best linear moment function, we should consider the conditional mean $\E\left(\mb{M}\mf{Y}^{**}_t,\mf{Z}^{**}_t|\mathcal{F}_{t-1}\right)$. The conditional mean of $\mf{Y}^{**}_{t-1}$ can be determined from
$\mf{Y}^{**,-1}_{t-1}=c_t\left(\mf{Y}^{*}_{t-1}-\frac{1}{T-t}\sum_{s=t}^{T-1}\mf{Y}^{*}_s\right)$, where $c_t=\left(\frac{T-t}{T-t+1}\right)^{1/2}$. Using Lemma~\ref{l4}, $\E\left(\mf{Y}^{**}_{t-1}|\mathcal{F}_{t-1}\right)$ can be approximated by 
\begin{align*}
\mf{H}_t&=c_t\left(\left(\mf{I}_n-\frac{1}{T-t}\sum_{h=1}^{T-t}\mf{A}^h\right)\mf{Y}^{*}_{t-1}- \frac{1}{T-t}\sum_{r=t}^{T-1}\left(\sum_{h=0}^{T-r-1}\mf{A}^{h}\right)\mf{S}^{-1}\left(\mf{X}_{r}\bs{\beta}_0+\alpha_{r,0}\mf{1}_n\right)\right)\nonumber\\
&-c_t\frac{1}{(T-t)(t-1)}\sum_{r=t}^{T-1}\left(\sum_{h=0}^{T-r-1}\mf{A}^{h}\right)\mf{S}^{-1}\sum_{s=1}^{t-1}\left(\mf{S}\mf{Y}^{*}_s-\mf{Z}^{*}_s\bs{\eta}_0-\alpha_{s0}\mf{1}_n\right),
\end{align*}
where $\mf{Z}^{*}_s=\left(\mf{Y}^{*}_{s-1},\mb{M}\mf{Y}^{*}_{s-1},\mf{X}_s\right)$. Thus, the best theoretical IV $\mf{J}_n\E\left(\mf{Y}^{**}_{t-1}|\mathcal{F}_{t-1}\right)$ can be approximated by $\mf{J}_n\mf{H}_t$.\footnote{Note that when $t=1$, we may simply use \linebreak $\mf{H}_1=c_1\left(\left(\mf{I}_n-\frac{1}{T-1}\sum_{h=1}^{T-1}\mf{A}^h\right)\mf{Y}^{*}_{0}- \frac{1}{T-1}\sum_{r=1}^{T-1}\left(\sum_{h=0}^{T-r-1}\mf{A}^{h}\right)\mf{S}^{-1}\left(\mf{X}_{r}\bs{\beta}_0+\alpha_{r,0}\mf{1}_n\right)\right)$.} Similarly, the best IVs for $\mf{J}\mf{Z}^{**}_t$ can be taken as $\mf{J}_n\mf{K}_t$, where $\mf{K}_t=\left(\mf{H}_t,\,\mb{M}\mf{H}_t,\,\mf{X}^{*}_t\right)$. Using \eqref{3.6}, the best IV for $\mf{J}_n\mf{M}_r\mf{Y}^{**}_t$ is $\mf{J}_n\mf{G}_r\left(\mf{K}_t\bs{\delta}_0+\alpha^{*}_{t0}\mf{1}_n\right)$ for $r=1,2,\hdots,p$. Overall, we may use $\mf{J}_n\mf{Q}^{*}_t$ as the IV matrix for $\mf{J}_n\left(\mb{M}\mf{Y}^{**}_t,\,\mf{Z}^{**}_t\right)$, where
\begin{align}\label{3.20}
\mf{Q}^{*}_t=\left(\mf{G}_1\left(\mf{K}_t\bs{\delta}_0+\alpha^{*}_{t0}\mf{1}_n\right),\hdots,\mf{G}_p\left(\mf{K}_t\bs{\delta}_0+\alpha^{*}_{t0}\mf{1}_n\right),\,\mf{K}_t\right),\quad t=1,2,\hdots,T-1.
\end{align}
The feasible version of $\mf{Q}^{*}_t$ can be obtained by substituting consistent estimators of the unknown parameters into \eqref{3.20}. Note that if $\mf{M}_r\mf{1}_n=\mf{1}_n$, i.e., when $\mf{M}_r$ is row-normalized, we have $\mf{J}_n\mf{M}_r=\mf{J}_n\mf{M}_r\mf{J}_n$. This property suggests that the time fixed effects will dropped from $\mf{J}_n\mf{Q}^{*}_t$ since $\mf{J}_n\mf{1}_n=\mf{0}_n$. However, if $\mf{M}_r$'s are not row normalized, then we also need an estimator of the time fixed effects to get a feasible version of $\mf{Q}^{*}_t$. Let $\hat{\bs{\vartheta}}_t=\mf{S}(\hat{\bs{\lambda}}_n)\mf{Y}^{*}_t-\mf{Z}^{*}_t\hat{\bs{\eta}}_n$ be an estimator of $\bs{\mu}_0+\mu_{\e}\mf{1}_n+\alpha_{t0}\mf{1}_n$. Under the normalization assumption of the form $\mf{1}^{'}_n\left(\bs{\mu}_0+\mu_{\e}\mf{1}_n\right)=\mf{0}_n$, we can estimate the time fixed effects by $\hat{\alpha}_t=\frac{1}{n}\mf{1}^{'}_n\hat{\bs{\vartheta}}_t$ for $t=1,2,\hdots,T$.\footnote{Note that when $T$ is large, $\tilde{\bs{\mu}}_0=\left(\bs{\mu}_0+\mu_{\e}\mf{1}_n\right)$ can be estimated by $\hat{\tilde{\bs{\mu}}}_n=\frac{1}{T}\sum_{t=1}^T\left(\hat{\bs{\vartheta}}_t-\frac{1}{n}\mf{1}^{'}_n\hat{\bs{\vartheta}}_t\mf{1}_n\right)$.} The following theorem provides our result on the best GMM estimator formulated with the feasible  versions of $\mf{Q}^{*}_t$ and $\mf{P}^{*}_j$ for $j=1,2,\hdots,p$.
\begin{thm}\label{t2}
Let $\hat{\mf{Q}}_t$ be the feasible version of $\mf{Q}^{*}_t$ for $t=1,2,\hdots,T-1$, and $\hat{\mf{P}}^{*}_j$ be the feasible version of $\mf{P}^{*}_j$ for $j=1,2,\hdots,p$. Consider the set of moment functions $\mf{g}_N(\bs{\theta})$ formulated with $\hat{\mf{Q}}_t$ and $\hat{\mf{P}}^{*}_j$. Then, the feasible best GMM estimator defined by $\hat{\bs{\theta}}^{*}_N=\argmin_{\bs{\theta}\in\bs{\Theta}}\mf{g}^{'}_N(\bs{\theta})\hat{\bs{\Omega}}^{-1}_N\mf{g}_N(\bs{\theta})$ has the following asymptotic distribution
\begin{align}\label{3.21}
\sqrt{N}\left(\hat{\bs{\theta}}^{*}_N-\bs{\theta}_0\right)\xrightarrow{d}N\left(\mf{0}_{k_z+p},\,\bs{\Sigma}^{*-1}_N\right),
\end{align}
where 
\begin{align}
\bs{\Sigma}^{*}_N=
\lim_{n,T\to\infty}
\begin{pmatrix}
\mf{C}^{*}_N/N&\mf{0}_{p\times k_z}\\
\mf{0}_{k_z\times p}&\mf{0}_{k_z\times k_z}
\end{pmatrix}
+\plim_{n,T\to\infty}\frac{1}{N\sigma^2_0}\left(\mf{L}_N,\, \mf{Z}_N\right)^{'}\mf{J}_N\left(\mf{L}_N,\, \mf{Z}_N\right),
\end{align}
with
\begin{align}
\mf{C}^{*}_{N}=
\begin{pmatrix}
\tr\left(\mf{G}^{'}_{1N}\mf{J}_N\mf{P}^{*s}_{1N}\mf{J}_N\right)&\hdots&\tr\left(\mf{G}^{'}_{pN}\mf{J}_N\mf{P}^{*s}_{1N}\mf{J}_N\right)\\
\vdots&\ddots&\vdots\\
\tr\left(\mf{G}^{'}_{1N}\mf{J}_N\mf{P}^{*s}_{pN}\mf{J}_N\right)&\hdots&\tr\left(\mf{G}^{'}_{pN}\mf{J}_N\mf{P}^{*s}_{pN}\mf{J}_N\right)
\end{pmatrix}.
\end{align}
\end{thm}
\begin{proof}
See  Section~\ref{pt2} of Appendix.
\end{proof}

\section{A Monte Carlo Study}\label{mc}
In this section, we investigate the finite sample properties of the best GMM estimator provided in Theorem~\ref{t2}. To that end, we consider $y_{it}=h_{it}^{1/2}\e_{it}$, and  the following cases for $h_{it}$:
\begin{align*}
&M_1:\, \log h_{it}=\sum_{j=1}^n\rho_{0}m_{ij}\log y^2_{jt}+\gamma_0\log y^2_{j,t-1}+\sum_{j=1}^n\delta_{0}m_{ij}\log y^2_{j,t-1} + \mathbf{x}^{'}_{it}\bs{\beta}_0 +\mu_{i0}+\alpha_{t0}\\
&M_2:\, \log h_{it}=\sum_{j=1}^n\rho_{0}m_{ij}\log y^2_{jt}+\gamma_0\log y^2_{j,t-1}+\sum_{j=1}^n\delta_{0}m_{ij}\log y^2_{j,t-1} +\mathbf{x}^{'}_{it}\bs{\beta}_0 + \mu_{i0},\\
&M_3:\, \log h_{it}=\sum_{l=1}^2\sum_{j=1}^n\rho_{l0}m_{l,ij}\log y^2_{jt}+\gamma_0\log y^2_{j,t-1}+\sum_{l=1}^2\sum_{j=1}^n\delta_{l0}m_{l,ij}\log y^2_{j,t-1}+\mathbf{x}^{'}_{it}\bs{\beta}_0+\mu_{i0}+\alpha_{t0},
\end{align*}
where $\mu_{i0}$'s and $\alpha_{t0}$'s are  i.i.d $N(0,1)$, and $\mf{x}_{it}\sim$ i.i.d $N(\mf{0}_{2\times1},\mf{I}_2)$ with $\bs{\beta}_0=(0.5,1)^{'}$. For the first two models, denoted $M_{1}$ and $M_{2}$, we consider two different temporal dependence structures -- a weakly temporal dependent model and a strongly persistent model. Specifically, we set $(\rho_0,\gamma_0,\delta_0)^{'}=\{(0.2,0.2, -0.2)^{'},\,(0.2,0.8,-0.2)^{'}\}$ in $M_{1}$ and $M_{2}$, respectively. Moreover, $M_2$ considers the case without temporal fixed effects, i.e., $\alpha_{t0} = 0$ for all $t$. In $M_3$, including higher-order spatial lags, we set $(\rho_{10},\rho_{20},\gamma_0,\delta_{10},\delta_{20})=(0.6,0.2,0.1,0.01,0.01)^{'}$. That is, we have very weak temporal and spatiotemporal effects. In all cases, we consider row-normalized queen contiguity spatial weights matrices, where $\mf{M}_1$ has positive weights for the first-lag neighbors and $\mf{M}_2$ for the second-lag neighbors. Furthermore, we consider two distributions to generate the disturbance terms: (i) $\e_{it}\sim$ i.i.d $N(0,1)$ and (ii) $\e_{it}\sim$ i.i.d $t_3$, where $t_3$ is the Student's $t$ distribution with $3$ degrees of freedom.  We set $(n,T)=\{(64,20),(100,40)\}$, and the number of repetitions to $1000$ in all cases. Thus, we considered 12 different model specifications in total.

The results of our Monte Carlo simulation study are reported in Tables \ref{table:mc1a} - \ref{table:mc2}. To evaluate the estimation performance, we report the average bias across all replications and the mean absolute errors (MAE). For all simulation settings, our theoretical findings are supported in the finite sample case. More precisely, when $n$ and $T$ increase, our suggested GMM estimator reports smaller bias and MAE in all cases. Comparing the performance with respect to the error distribution, we see slightly lower MAEs in the heavy-tailed case. These differences are insignificant in almost all cases ($\alpha = 0.05$). Overall, these results indicate that our suggested GMM estimator has good finite sample properties in terms of bias and MAE.

\begin{table}
\caption{Average bias and mean absolute errors (MAE) of the estimated parameters of Model $M_{1}$.}\label{table:mc1a}
\begin{center}
\spacingset{1} 
\begin{tabular}{l l c c c c}
\hline
		& & \multicolumn{2}{c}{$\e_{it}\sim$ i.i.d $N(0,1)$} & \multicolumn{2}{c}{$\e_{it}\sim$ i.i.d $t_3$} \\
		& & $n = 64$ & $n = 100$                             & $n = 64$ & $n = 100$                          \\
		& & $T = 20$ & $T = 40$                              & $T = 20$ & $T = 40$                           \\
	\hline
		& $\rho_0 = 0.2$     &  0.0034 &  0.0041 & -0.0001 &  0.0027 \\
		& $\gamma_0 = 0.2$   &  0.0001 & -0.0008 & -0.0004 &  0.0005 \\
Bias	        & $\delta_0 = -0.2$  & -0.0013 & -0.0004 &  0.0007 & -0.0009 \\
		& $\beta_{00} = 0.5$ & -0.0033 & -0.0023 & -0.0020 & -0.0003 \\
		& $\beta_{10} = 1$   & -0.0056 & -0.0007 & -0.0026 & -0.0035 \\[.1cm]
		& $\rho_0 = 0.2$     &  0.1142 &  0.0590 &  0.1183 &  0.0653 \\ 
		& $\gamma_0 = 0.2$   &  0.0266 &  0.0139 &  0.0253 &  0.0134 \\
MAE 	        & $\delta_0 = -0.2$  &  0.0612 &  0.0321 &  0.0606 &  0.0348 \\ 
		& $\beta_{00} = 0.5$ &  0.0527 &  0.0292 &  0.0581 &  0.0324 \\ 
		& $\beta_{10} = 1$   &  0.0517 &  0.0287 &  0.0573 &  0.0318 \\
	\hline
\end{tabular}
\spacingset{1.8} 
\end{center}
\end{table}

\begin{table}
\caption{Average bias and mean absolute errors (MAE) of the estimated parameters of Model $M_{2}$.}\label{table:mc1b}
\begin{center}
\spacingset{1} 
\begin{tabular}{l l c c c c}
\hline
		& & \multicolumn{2}{c}{$\e_{it}\sim$ i.i.d $N(0,1)$} & \multicolumn{2}{c}{$\e_{it}\sim$ i.i.d $t_3$} \\
		& & $n = 64$ & $n = 100$                             & $n = 64$ & $n = 100$                          \\
		& & $T = 20$ & $T = 40$                              & $T = 20$ & $T = 40$                           \\
	\hline
		& $\rho_0 = 0.2$     &  0.0214 &  0.0119 &  0.0204 &  0.0035 \\
		& $\gamma_0 = 0.8$   & -0.0014 & -0.0012 & -0.0025 & -0.0013 \\
Bias	        & $\delta_0 = -0.2$  &  0.0109 & -0.0019 & -0.0036 &  0.0033 \\
		& $\beta_{00} = 0.5$ & -0.0031 &  0.0001 & -0.0039 & -0.0014 \\
		& $\beta_{10} = 1$   & -0.0054 & -0.0018 & -0.0091 &  0.0002 \\[.1cm]
		& $\rho_0 = 0.2$     &  0.1096 &  0.0582 &  0.1161 &  0.0623 \\ 
		& $\gamma_0 = 0.8$   &  0.0372 &  0.0167 &  0.0383 &  0.0157 \\
MAE 	        & $\delta_0 = -0.2$  &  0.1184 &  0.0627 &  0.1238 &  0.0643 \\ 
		& $\beta_{00} = 0.5$ &  0.0526 &  0.0296 &  0.0553 &  0.0321 \\ 
		& $\beta_{10} = 1$   &  0.0532 &  0.0276 &  0.0581 &  0.0308 \\
	\hline
\end{tabular}
\spacingset{1.8} 
\end{center}
\end{table}

\begin{table}
\caption{Average bias and mean absolute errors (MAE) of the estimated parameters of Model $M_3$.}\label{table:mc2}
\begin{center}
\spacingset{1} 
\begin{tabular}{l l c c c c}
\hline
		& & \multicolumn{2}{c}{$\e_{it}\sim$ i.i.d $N(0,1)$} & \multicolumn{2}{c}{$\e_{it}\sim$ i.i.d $t_3$} \\
		& & $n = 49$ & $n = 100$                             & $n = 49$ & $n = 100$                          \\
		& & $T = 20$ & $T = 40$                              & $T = 20$ & $T = 40$                           \\
	\hline
		& $\rho_{10} = 0.6$     &  0.0112 &  0.0026 &  0.0150 &  0.0037 \\
		& $\rho_{20} = 0.2$     &  0.0122 &  0.0041 &  0.0134 &  0.0039 \\
		& $\gamma_0 = 0.1$      & -0.0007 & -0.0004 & -0.0008 &  0.0003 \\
Bias	        & $\delta_{10} = 0.01$  & -0.0066 & -0.0014 & -0.0076 & -0.0012 \\
		& $\delta_{20} = 0.01$  & -0.0059 & -0.0023 & -0.0070 & -0.0037 \\
		& $\beta_{00} = 0.5$    & -0.0006 &  0.0022 & -0.0043 & -0.0013 \\
		& $\beta_{10} = 1$      & -0.0064 & -0.0046 & -0.0059 & -0.0062 \\[.1cm]
		& $\rho_{10} = 0.6$     &  0.0794 &  0.0458 &  0.0864 &  0.0479 \\
		& $\rho_{20} = 0.2$     &  0.1183 &  0.0650 &  0.1271 &  0.0659 \\
		& $\gamma_0 = 0.1$      &  0.0254 &  0.0136 &  0.0261 &  0.0135 \\
MAE 	        & $\delta_{10} = 0.01$  &  0.0514 &  0.0277 &  0.0525 &  0.0279 \\
		& $\delta_{10} = 0.01$  &  0.0646 &  0.0338 &  0.0648 &  0.0351 \\
		& $\beta_{00} = 0.5$    &  0.0559 &  0.0293 &  0.0556 &  0.0313 \\ 
		& $\beta_{10} = 1$      &  0.0534 &  0.0306 &  0.0599 &  0.0337 \\
	\hline
\end{tabular}
\spacingset{1.8} 
\end{center}
\end{table}

\section{Real-World Example: Intra-city housing market risk}\label{emp}

The real-estate market is undoubtedly a financial market with the most apparent spatial and temporal dependence. The location of a property, along with size and condition, is an important price-determining influence. Hence, there are pronounced spatial spillover effects in real-estate prices, in addition to the natural temporal dependence. Furthermore, taxes may significantly affect the market, such as property taxes or real-estate transfer taxes. In general, however, taxes appear to play a subordinate role in purchasing decisions -- with one exception, namely, if the transfer taxes change, some sales could be shifted for a certain period. If, for example, the land transfer tax increases by one percentage point and one wants to buy a property in January, it is profitable to conclude the purchase contract already in December. This results in a shift of property sales from January to December, and thus more sales in December and fewer sales in January than expected. However, does this also impact the risk of the real-estate market?

\begin{figure}
	\begin{center}
		\includegraphics[width = 0.40\textwidth]{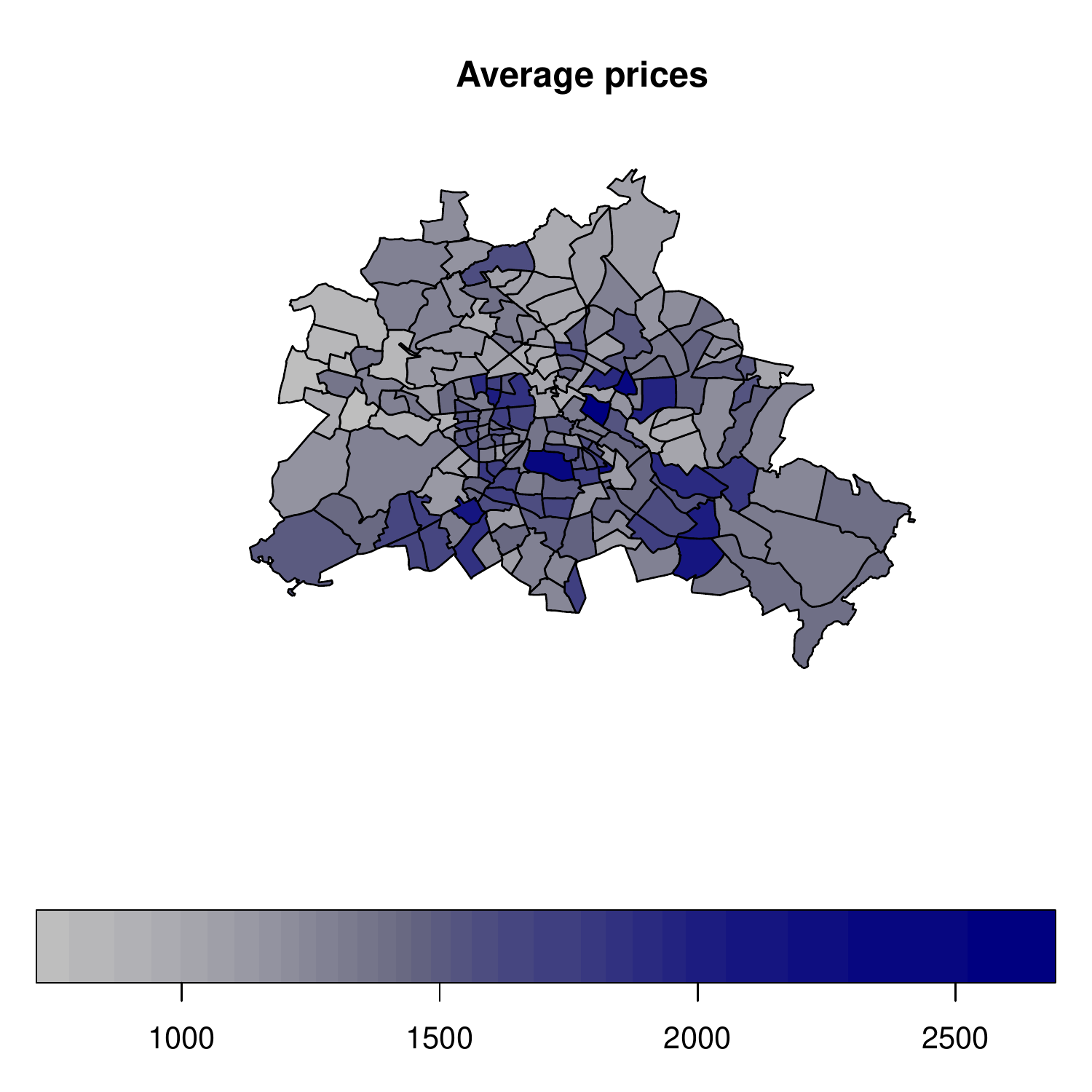}
		\includegraphics[width = 0.56\textwidth]{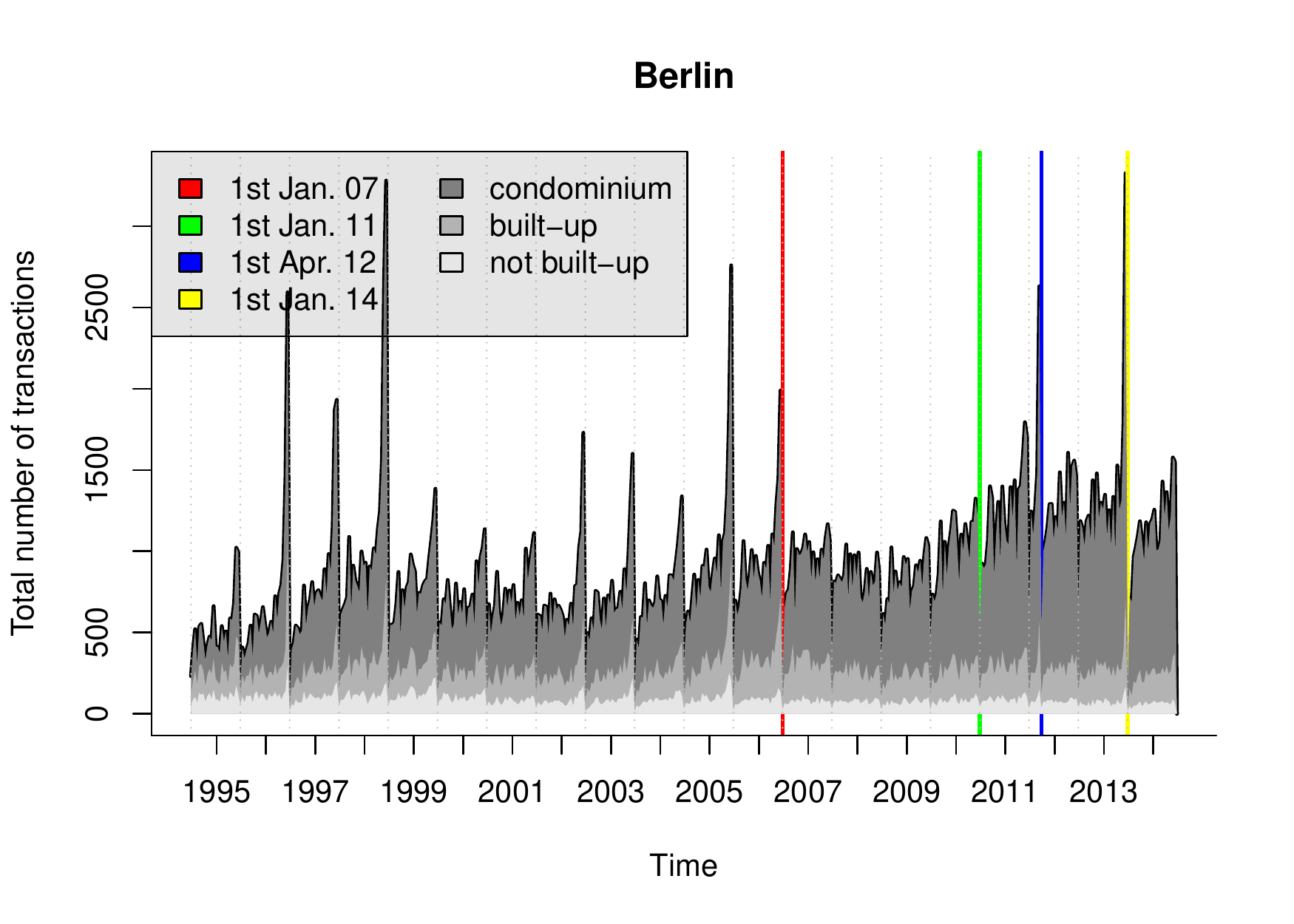}
	\end{center}
	\caption{Overview of the data set. Left: Average house prices for each of the 190 zip-code areas over the period from January 1995 to December 2015. Right: Total number of real-estate transactions in all zip-code areas of Berlin. The vertical bars indicate the time points of changes in the real-estate transfer taxes in Berlin or the surrounding state of Brandenburg (red: increase from 3.5 to 4.5 \% in Berlin; green: increase from 3.5 to 5 \% in Brandenburg; blue: increase from 4.5 to 5 \% in Berlin; yellow: increase from 5 to 6 \% in Berlin)}\label{fig:total}	
\end{figure}

For the empirical analysis, we use monthly log-returns of the average sales prices of all condominium sales in all postcode regions of Berlin from January 1995 to December 2015 (see Figure \ref{fig:total}, left). The relative price per square meter is determined for each zip-code area from an average of 6.31 sales per month. The specific location of the German capital Berlin in the centre of another federal state, Brandenburg, makes it a very intriguing example. The surrounding area of Berlin is very well connected to the city center by public transport and infrastructure, such that exogenous effects such as tax changes in Brandenburg may have an impact on Berlin and vice versa. Every real estate purchase in Germany is subject to the real estate transfer tax, which must be paid once at the time of purchase. The amount of tax depends on the purchase price. Until 1.9.2006, a unified tax rate of 3.5 per cent was applied in Germany as a whole. Afterwards, each federal state could set its tax rate, and there were gradual increases in all federal states. Specifically, Berlin increased the tax rates from 3.5 to 4.5 per cent on 1.1.2007, from 4.5 to 5 per cent on 1.4.2012, and from 5 to 6 per cent on 1.1.2014. The shifting effects described above can also be observed for Berlin, as it is shown in Figure \ref{fig:total}. In addition, an end-of-year effect is clearly visible due to other accounting and tax reasons. This motivates why we would expect different market risks at the end and beginning of a year. To estimate these temporal effects, we consider a model without temporal fixed effects (i.e., $\alpha_{t0} = 0$ for all $t$ in \eqref{2.2}) and model the temporal effects by including yearly and monthly indicator variables as regressors.

\begin{table}
\caption{Estimated parameters of the dynamic ARCH process and diagnostic measures of the residuals.}\label{table:results_berlin}
\begin{center}
\spacingset{1} 
\begin{tabular}{c p{5.5cm}  c r r r}
&    & Parameter & Estimate & Standard error & T-statistics \\
	\hline
\multicolumn{6}{l}{\emph{Regressive effects}} \\
& Total number of transactions                          & $\beta_1$    &       -0.1356   &      0.0184  &  -7.3509  \\
\multicolumn{6}{l}{\emph{End-of-year effects}} \\
& December                                              & $\beta_2$    &       -0.0516   &      0.0465  &  -1.1098  \\
& January                                               & $\beta_3$    &        0.2152   &      0.0729  &   2.9527  \\
& February                                              & $\beta_4$    &        0.1465   &      0.0555  &   2.6379  \\
\multicolumn{6}{l}{\emph{Yearly effects}} \\
& 1997                                                  & $\beta_5$    &       -0.1447   &      0.0603  &  -2.4002  \\
& 1998                                                  & $\beta_6$    &       -0.2112   &      0.0735  &  -2.8739  \\
& 1999                                                  & $\beta_7$    &       -0.2003   &      0.0708  &  -2.8292  \\
& 2000                                                  & $\beta_8$    &       -0.1136   &      0.0585  &  -1.9413  \\
& 2001                                                  & $\beta_9$    &       -0.1194   &      0.0600  &  -1.9897  \\
& 2011                                                  & $\beta_{10}$ &       -0.1244   &      0.0674  &  -1.8466  \\
& 2012                                                  & $\beta_{11}$ &       -0.1161   &      0.0684  &  -1.6974  \\
& 2013                                                  & $\beta_{12}$ &       -0.1204   &      0.0684  &  -1.7606  \\
\multicolumn{6}{l}{\emph{Tax effects}} \\
& Month before tax increase                             & $\beta_{13}$ &       -0.0559   &      0.1120  &  -0.4988  \\
\multicolumn{6}{l}{\emph{Spatiotemporal effects}} \\
& Spatial interaction (contiguity-based)                & $\rho$       &        0.4032   &      0.1450  &  2.7801  \\
& Temporal interaction (first time lag)                 & $\gamma$     &        0.1913   &      0.0051  & 37.4971  \\
& Spatiotemporal interaction                            & $\delta$     &       -0.0737   &      0.0313  & -2.3547  \\
\multicolumn{6}{l}{\emph{Model diagnostics}} \\
& BIC                                                                                    & & 82572.28  & &  \\
& n                                                                                      & & 190  & &  \\
& T                                                                                      & & 239  & &  \\
& {Percentage of locations with significantly (temp.) autocorrelated errors ($\alpha = 5$ \%)}             & & 18 & &  \\
& {Percentage of time points with significantly (spatially) positive autocorrelated errors (Moran's $I$, $\alpha = 5$ \%)} & & 0 & &  \\
	\hline
\end{tabular}
\spacingset{1.8} 
\end{center}
\end{table}

We estimate a first-order version of our model in \eqref{2.2}. We specify the spatial weights matrix (row-standardized) based on the queen contiguity scheme, where all adjacent neighbors are equally weighted. In Table \ref{table:results_berlin} and Figure \ref{fig:est_h}, we present the estimated parameters of our dynamic spatiotemporal ARCH model and the estimated volatility, respectively.  The included covariates were selected by stepwise excluding regressors, such that the Bayesian information criterion is minimized. The overall market dynamic measured by the total numbers of real-estate transactions has the most significant effect on the volatility (see Table \ref{table:results_berlin}). The more transactions, the lower the log-volatility. In addition, we observe lower volatilities at the end of each year, and significantly higher volatilities in January and February. These effects decrease from January to February ($0.2152$ to $0.1465$). The March effects were already insignificant. Further, we observe two periods of significantly lower risks compared to the remaining periods, namely 1997-2001 and 2011-2013. The anticipated tax effects are not significant though. Thus, we could not find evidence that the legal changes in the taxation framework affect the volatility of log-returns.

As expected, we also see significant spatial and spatiotemporal spill-over effects. The spatial ARCH parameter $\hat{\rho} = 0.4032$ is on moderate level.  That is, an increase in the log-squared return in one location instantaneously increases the log-volatility in the adjacent regions. Further, the temporal dependence is composed of a purely temporal lag ($\hat{\gamma} = 0.1913$) and a spatiotemporal lag ($\hat{\delta} = -0.0737$), which in total indicate a moderate temporal persistence. 

When analyzing the estimated volatility, we clearly see temporal patterns with reduced risks during the two above-mentioned periods. More interestingly, there are several regions of higher volatility, mostly located at the outer zip-code regions in the North and North-West, as shown in the top panels of Figure \ref{fig:est_h}, where the averaged volatility estimates are depicted. Moreover, the average volatility over all zip-codes  changes over time as shown in the first figure of the top panels of Figure \ref{fig:est_h}. Finally, to illustrate the estimated $h_{it}$'s for one selected location, we show the results for Berlin-Tempelhof (zip-code of the closed airport Berlin-Tempelhof) in April 2012, the month after the increase of the real-estate transfer taxes from 4.5 to 5 \%. We do not see different patterns in the volatility estimates after the airport was closed, marked by the dashed black line.

\begin{figure}
	\begin{center}
		\includegraphics[width = 0.45\textwidth]{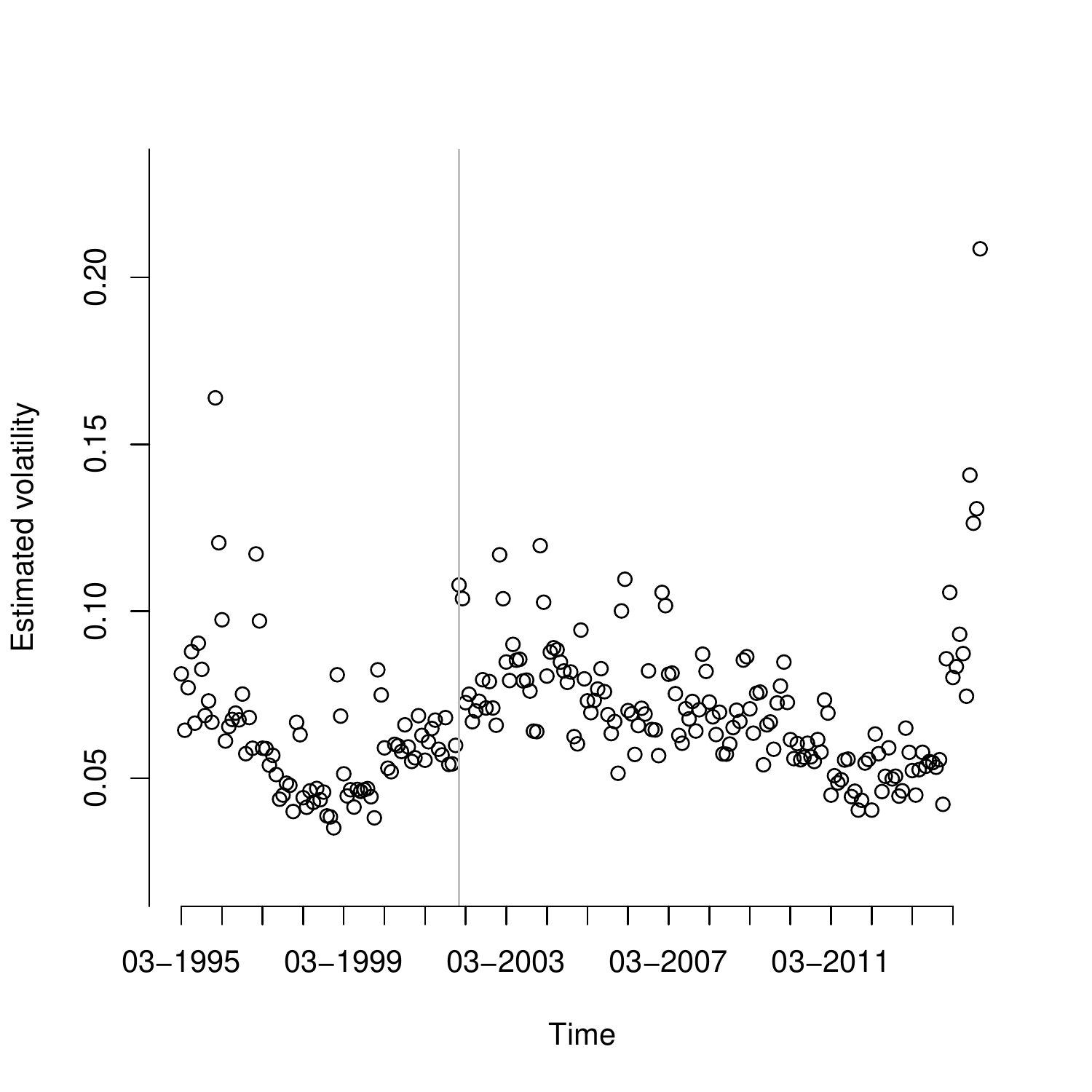}
		\includegraphics[width = 0.45\textwidth]{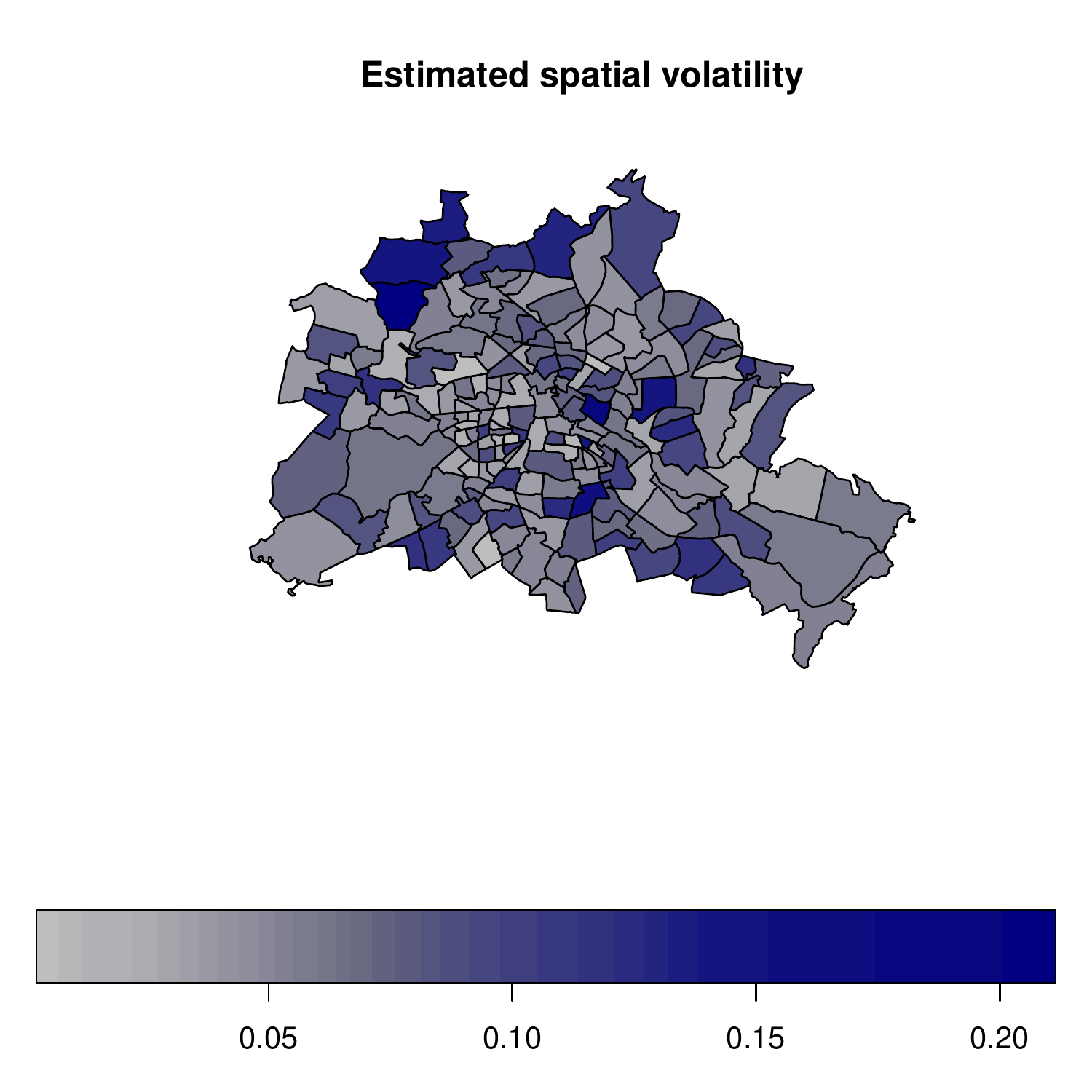}\\
		\includegraphics[width = 0.45\textwidth]{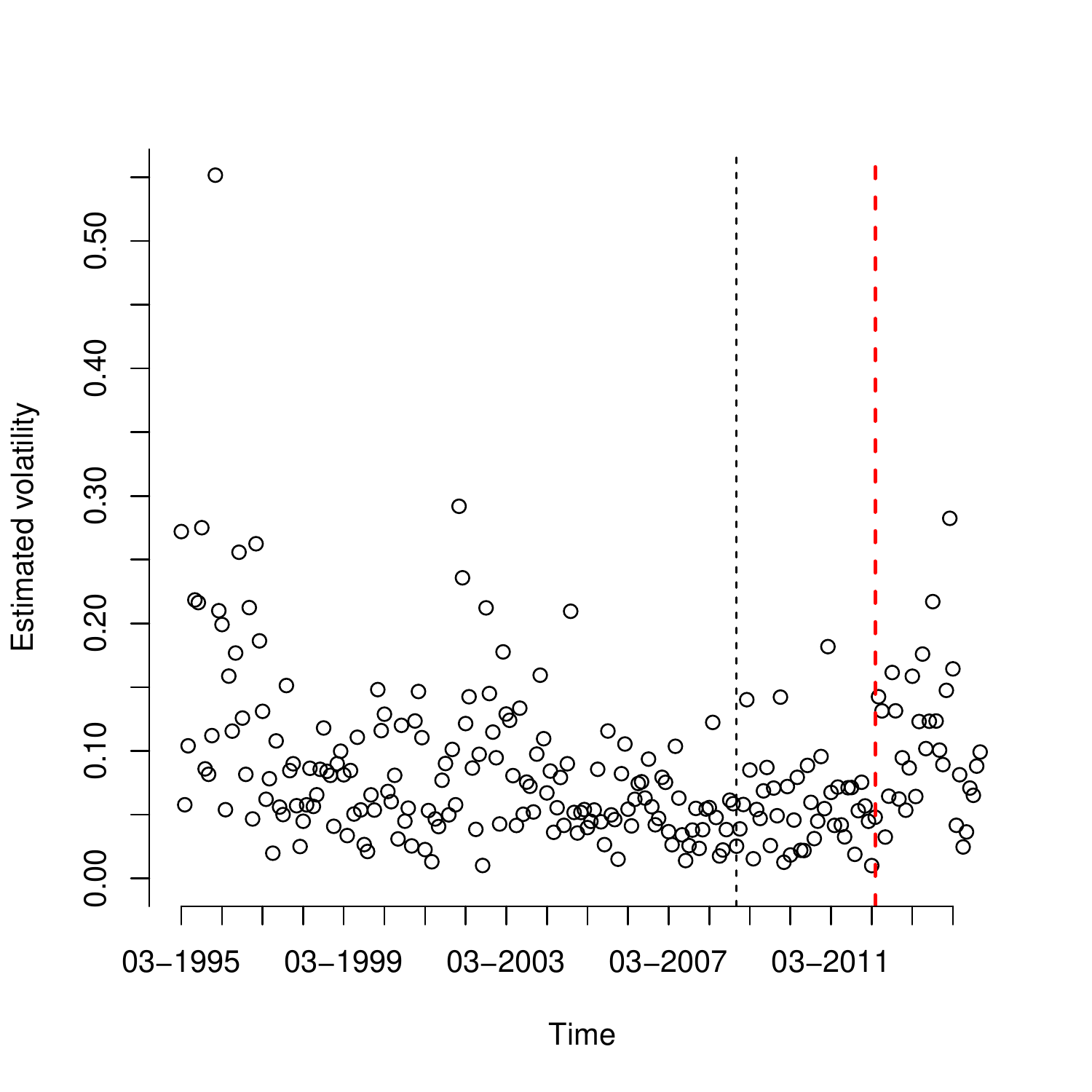}
		\includegraphics[width = 0.45\textwidth]{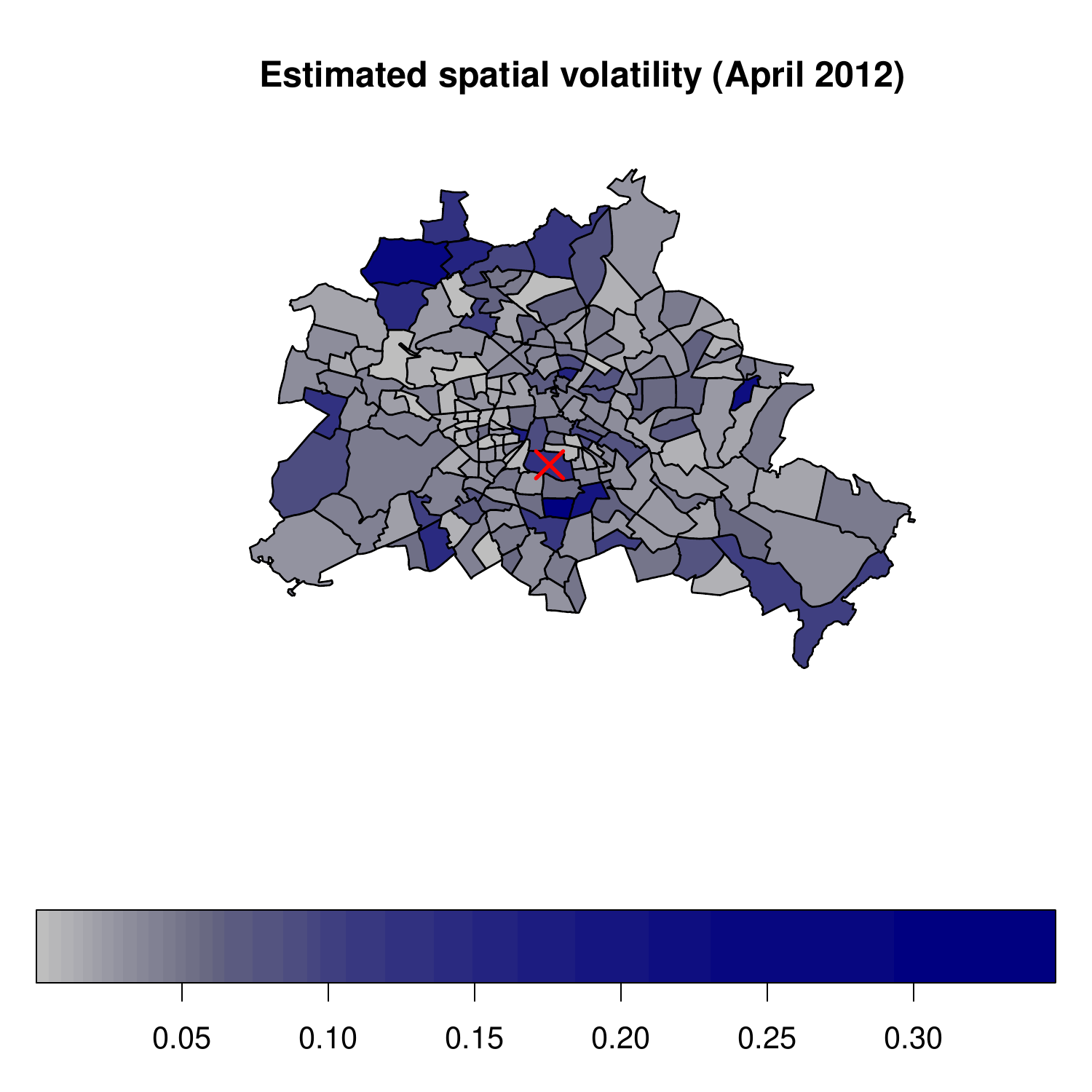}\\
	\end{center}
	\caption{Estimated conditional volatility $\hat{h}_{it}$. Top left: The spatially averaged volatilities levels are displayed over time $\frac{1}{n}\sum_{i = 1}^{n}\hat{h}_{it}$. The vertical grey line marks the Euro introduction in January 2002. Top right: The temporally averaged volatilities levels are displayed on the map $\frac{1}{T}\sum_{t = 1}^{T}\hat{h}_{it}$. Bottom left: $\hat{h}_{it}$ of one selected region over time (Berlin Tempelhof, zip-code of the former airport, marked by red cross in the bottom right plot). Bottom right: $\hat{h}_{it}$ at one selected time point (April 2012, first month after increase of real-estate transfer taxes from 4.5 to 5 \%, shown by dashed red line in the bottom-left plot).}\label{fig:est_h}
\end{figure}

\section{Conclusion}\label{conc}

In this paper, we introduced a dynamic spatiotemporal ARCH model that allows for unobserved heterogeneity over time and space. The model can be used to describe the spatiotemporal clustering effect in the volatility of a random process. As typically observed for spatial data, the model allows for instantaneous spill-over effects across space, which is the main difference to multivariate time-series GARCH models. In the latter case, spatial interactions would only occur after one time lag. In addition to these instantaneous spatial effects, the model includes temporal and spatiotemporal autoregressive effects of the log-squared returns in the log-volatility equation. While the temporal effect measures the dependence between the current and past observation of the same spatial unit, the spatiotemporal coefficients describe the dependence between the observation in one location and its past observations at neighboring locations.

For our suggested dynamic spatiotemporal ARCH model, we obtain an estimation equation by applying a log-square transformation together with an orthonormal and a deviation from group-mean operator to eliminate the fixed effects. We introduced a GMM  estimation approach based on a set of linear and quadratic moment functions of the transformed process. We establish the consistency and asymptotic normality of our suggested GMM estimator under fairly general assumptions for large and finite $T$ cases. Moreover, when the number of time periods is large, we present an optimal set of moment functions that leads to an efficient estimator.

We investigated the finite-sample performance of our suggested estimator in a series of Monte-Carlo simulations under different model settings and error distributions. Overall, the simulation results are in line with our theoretical claims. In an empirical application, we illustrated the use of our model for the log-returns of the intra-city real-estate prices in Berlin over the period 1995 - 2015. Our estimation results  show that the spatial, temporal and spatiotemporal lags of the log-squared returns have statistically significant effect on the log-volatility. This leads to temporal and spatial spill-over effects. We showed that the average volatility of log-returns over space and time varies significantly. Finally, our model allows us to estimate the market risk in terms of the volatility in each location and time point.

In future studies, our model can be extended in a number of ways. First, we considered additive time and space fixed effects in the log-volatility equation. Instead of this additive structure, a log-volatility equation that includes interactive fixed effects can be studied. Second, the spatial and spatiotemporal lags in the log-volatility equation can be formulated with time-varying spatial weights matrices. Finally, we can also allow for potential endogeneity in the spatial weights instead of exogenous spatial weights. All of these extensions can be explored in future studies.

\newpage
\appendix
\noindent{\textbf{\Large{Appendix}}}
\renewcommand\thefigure{\thesection.\arabic{figure}}   
\renewcommand\thetable{\thesection.\arabic{table}}   
\section{Some Useful Lemmas}\label{app}
In this section, we provide four lemmas that are essential for our main results. 
\begin{lemma}\label{l1}
Let $\mathcal{F}_{t-1}$ be the $\sigma$-algebra generated by $\left(\mf{Y}_0,\hdots,\mf{Y}_{t-1}\right)$ conditional on $\left(\mf{X}_1,\hdots,\mf{X}_T,\bs{\mu}_0,\bs{\alpha}_0\right)$, and $\mf{B}$ be a non-stochastic $n\times n$ matrix that has row and column sums uniformly bounded in absolute value. Then, under Assumption~\ref{a1}, we have the following results.
\begin{enumerate}
\item $\E\left(\mf{U}^{*'}_t\mf{B}\mf{U}^{*}_s|\mathcal{F}_{t-1}\right)=0$ for $t\ne s$.
\item $\frac{1}{N}\mf{U}^{'}_N\mf{B}_N\mf{U}_N=\frac{\sigma^2_0}{n}\tr\left(\mf{B}\right)+O_p\left(N^{-1/2}\right)=O_p(1)$, where $\mf{B}_N=\left(\mf{I}_{T-1}\otimes\mf{B}\right)$.
\item $\E\left(\left(\mf{U}^{*'}_t\mf{B}\mf{U}^{*}_t\right)^2\right)=\left(\mu_4-3\sigma^4_0\right)c^4_{t}\left(1+\frac{1}{(T-t)^3}\right)\vec^{'}_D(\mf{B})\vec_D(\mf{B})+\sigma^4_0\left(\tr^2(\mf{B})+\tr\left(\mf{B}\mf{B}^s\right)\right)$, where $c_{t}=\left(\frac{T-t}{T-t+1}\right)^{1/2}$.
\item Under Assumption~\ref{a6}, we have $\plim_{n\to\infty}\frac{1}{N}\sum_{t=1}^{T-1}\mf{Q}^{'}_t\mf{B}\mf{U}^{*}_t=\mf{0}$, where $\mf{Q}_t$ is the IV matrix.
\item Under Assumption~\ref{a5}, we have $\frac{1}{N}\mf{Y}^{'}_{N,-1}\mf{B}_N\mf{U}_N-\E\left(\frac{1}{N}\mf{Y}^{'}_{N,-1}\mf{B}_N\mf{U}_N\right)=O_p\left(N^{-1/2}\right)$, where $\E\left(\frac{1}{N}\mf{Y}^{'}_{N,-1}\mf{B}_N\mf{U}_N\right)=\frac{\sigma^2_0}{N}\tr\left(\mf{B}^{'}\left(\sum_{h=1}^{T-1}\left(1-\frac{h}{T}\right)\mf{A}^{h-1}\right)\mf{S}^{-1}\right)$ is $O(T^{-1})$.
\end{enumerate}
\end{lemma}
\begin{proof}
See Lemma 2 in \citet{Lee:2014}. 
\end{proof}
\begin{lemma}\label{l2}
Under Assumption~\ref{a1}, we have the following results.
\begin{enumerate}
\item For $l,j=1,2,\hdots,m$, we have
\begin{align*}
&\E\left(\mf{U}^{'}_N\mf{J}_N\mf{P}_{lN}\mf{J}_N\mf{U}_N\cdot\mf{U}^{'}_N\mf{J}_N\mf{P}_{jN}\mf{J}_N\mf{U}_N\right)\\
&=\sigma^4_0\tr\left(\mf{J}_N\mf{P}_{lN}\mf{J}_N\left(\mf{J}_N\mf{P}_{jN}\mf{J}_N\right)^s\right)+(\mu_4-3\sigma^4_0)\vec^{'}_D\left(\mf{J}_N\mf{P}_{lN}\mf{J}_N\right)\vec_D\left(\mf{J}_N\mf{P}_{jN}\mf{J}_N\right).
\end{align*}

\item For $l=1,2,\hdots,m$, we have
\begin{align*}
\E\left(\mf{Q}^{'}_N\mf{J}_N\mf{U}_N\cdot \mf{U}^{'}_N\mf{J}_N\mf{P}_{lN}\mf{J}_N\mf{U}_N\right)=\mf{0}_{q\times1}.
\end{align*}
\end{enumerate}
\end{lemma}
\begin{proof}
This lemma is a simple extension of Lemma 2 of \citet{Lee:2014} to our setting. Therefore, we omit its proof.
\end{proof}
\begin{lemma}\label{l3}
Let $\mf{p}_t$ be $n\times1$ column vector from  the IV matrix $\mf{Q}_t$ given in Assumption~\ref{a6}. Consider $a_t=\mf{q}^{'}_t\mf{U}^{*}_t+\mf{U}^{*'}_t\mf{B}\mf{U}^{*}_t-\sigma^2_0\tr(\mf{B})$, where $\mf{B}$ is a non-stochastic $n\times n$ matrix that has row and column sums uniformly bounded in absolute value. Using Lemma~\ref{l2}, we have $\var\left(\sum_{t=1}^{T-1}a_t\right)=\sigma^2_0\sum_{t=1}^{T-1}\mf{q}^{'}_t\mf{q}_t+T\left(\mu_4-3\sigma^4_0\right)\sum_{i=1}^nb^2_{ii}+T\sigma^4_0\tr\left(\mf{B}\mf{B}^s\right)$. If $\left\{\frac{1}{N}\var\left(\sum_{t=1}^{T-1}a_t\right)\right\}$ is bounded away from zero, then 
$$
\frac{\sum_{t=1}^{T-1}a_t}{\var^{1/2}\left(\sum_{t=1}^{T-1}a_t\right)}\xrightarrow{d}N(0,\,1).
$$
\end{lemma}
\begin{proof}
See the CLT results in \citet{Lee:2014,Yu:2008}.
\end{proof}
\begin{lemma}\label{l4}
Let  $\mf{Z}^{*}_s=\left(\mf{Y}^{*}_{s-1},\mb{M}\mf{Y}^{*}_{s-1},\mf{X}_s\right)$ and $c_t=\left(\frac{T-t}{T-t+1}\right)^{1/2}$. Then,
\begin{align} 
\mf{Y}^{**,-1}_{t-1}&=c_t\left(\left(\mf{I}_n-\frac{1}{T-t}\sum_{h=1}^{T-t}\mf{A}^h\right)\mf{Y}^{*}_{t-1}- \frac{1}{T-t}\sum_{r=t}^{T-1}\left(\sum_{h=0}^{T-r-1}\mf{A}^{h}\right)\mf{S}^{-1}\left(\mf{X}_{r}\bs{\beta}_0+\alpha_{r,0}\mf{1}_n\right)\right)\nonumber\\
&-c_t\frac{1}{(T-t)(t-1)}\sum_{r=t}^{T-1}\left(\sum_{h=0}^{T-r-1}\mf{A}^{h}\right)\mf{S}^{-1}\sum_{s=1}^{t-1}\left(\mf{S}\mf{Y}^{*}_s-\mf{Z}^{*}_s\bs{\eta}_0-\alpha_{s0}\mf{1}_n\right)\\
&-c_t\frac{1}{(T-t)(t-1)}\sum_{r=t}^{T-1}\left(\sum_{h=0}^{T-r-1}\mf{A}^{h}\right)\mf{S}^{-1}\sum_{s=1}^{t-1}\mf{U}_s-c_t\frac{1}{T-t}\sum_{r=t}^{T-1}\left(\sum_{h=0}^{T-r-1}\mf{A}^{h}\right)\mf{S}^{-1}\mf{U}_{r}.\nonumber
\end{align}
\end{lemma}
\begin{proof}
Using \eqref{3.1}, we have
\begin{align}\label{B.1}
\mf{Y}^{*}_t&=\mf{A}\mf{Y}^{*}_{t-1}+\mf{S}^{-1}\left(\mf{X}_t\bs{\beta}_0+\tilde{\bs{\mu}}_0+\alpha_{t0}\mf{1}_n+\mf{U}_t\right),
\end{align}
where $\tilde{\bs{\mu}}_0=\bs{\mu}+\mu_{\e}\mf{1}_n$. Using \eqref{B.1}, we can expand $\mf{Y}^{*}_{t+h}$ for $h\geq0$ as
\begin{align}
\mf{Y}^{*}_{t+h}&=\mf{A}^{h+1}\mf{Y}^{*}_{t-1}+\sum_{j=0}^h\mf{A}^{j}\mf{S}^{-1}\left(\mf{X}_{t+h-j}\bs{\beta}_0+\tilde{\bs{\mu}}_0+\alpha_{t+h-j,0}\mf{1}_n+\mf{U}_{t+h-j}\right),
\end{align} 
Therefore, we have
\begin{align}
\mf{Y}^{*}_{s}&=\mf{A}^{s-t+1}\mf{Y}^{*}_{t-1}+\sum_{j=0}^{s-t}\mf{A}^{j}\mf{S}^{-1}\left(\mf{X}_{s-j}\bs{\beta}_0+\tilde{\bs{\mu}}_0+\alpha_{s-j,0}\mf{1}_n+\mf{U}_{s-j}\right).
\end{align}
Then, we can express $\sum_{s=t}^{T-1}\mf{Y}^{*}_{s}$ as
\begin{align}
\sum_{s=t}^{T-1}\mf{Y}^{*}_{s}&=\sum_{s=t}^{T-1}\mf{A}^{s-t+1}\mf{Y}^{*}_{t-1}+\sum_{s=t}^{T-1}\sum_{j=0}^{s-t}\mf{A}^{j}\mf{S}^{-1}\left(\mf{X}_{s-j}\bs{\beta}_0+\tilde{\bs{\mu}}_0+\alpha_{s-j,0}\mf{1}_n+\mf{U}_{s-j}\right)\nonumber\\
&=\sum_{h=1}^{T-t}\mf{A}^{h}\mf{Y}^{*}_{t-1}+\sum_{r=t}^{T-1}\left(\sum_{h=0}^{T-r-1}\mf{A}^{h}\right)\mf{S}^{-1}\left(\mf{X}_{r}\bs{\beta}_0+\tilde{\bs{\mu}}_0+\alpha_{r,0}\mf{1}_n+\mf{U}_{r}\right).
\end{align}
Thus, using $\mf{Y}^{**,-1}_{t-1}=c_t\left(\mf{Y}^{*}_{t-1}-\frac{1}{T-t}\sum_{s=t}^{T-1}\mf{Y}^{*}_s\right)$, we obtain
\begin{align}\label{B.5}
\mf{Y}^{**,-1}_{t-1}&=c_t\left(\left(\mf{I}_n-\frac{1}{T-t}\sum_{h=1}^{T-t}\mf{A}^h\right)\mf{Y}^{*}_{t-1}- \frac{1}{T-t}\sum_{r=t}^{T-1}\left(\sum_{h=0}^{T-r-1}\mf{A}^{h}\right)\mf{S}^{-1}\left(\mf{X}_{r}\bs{\beta}_0+\alpha_{r,0}\mf{1}_n\right)\right)\nonumber\\
&-c_t\frac{1}{T-t}\sum_{r=t}^{T-1}\left(\sum_{h=0}^{T-r-1}\mf{A}^{h}\right)\mf{S}^{-1}\tilde{\bs{\mu}}_0-c_t\frac{1}{T-t}\sum_{r=t}^{T-1}\left(\sum_{h=0}^{T-r-1}\mf{A}^{h}\right)\mf{S}^{-1}\mf{U}_{r}.
\end{align}
Note that we can express $\tilde{\bs{\mu}}_0$ as 
\begin{align}\label{B.6}
\tilde{\bs{\mu}}_0=\frac{1}{t-1}\sum_{s=1}^{t-1}\left(\mf{S}\mf{Y}^{*}_s-\mf{Z}^{*}_s\bs{\eta}_0-\alpha_{s0}\mf{1}_n-\mf{U}_s\right),
\end{align}
where $\mf{Z}^{*}_s=\left(\mf{Y}^{*}_{s-1},\mb{M}\mf{Y}^{*}_{s-1},\mf{X}_s\right)$. Then, substituting \eqref{B.6} into \eqref{B.5} yields the result.
\end{proof}
\section{Details on Identification }\label{A.C}
In the GMM setting, the parameter vector is identified if $\plim_{n\to\infty}\frac{1}{N}\mf{g}_N(\bs{\theta})=\mf{0}$ has the unique solution $\bs{\theta}_0$. Consider the linear moment function $\mf{Q}^{'}_N\mf{J}_N\mf{U}_N(\bs{\theta})$. From \eqref{3.3}, we have $\mf{J}_n\mf{U}^{*}_t(\bs{\theta})=\mf{J}_n\left(\mf{S}(\bs{\rho})\mf{Y}^{**}_t-\mf{Z}^{**}_t\bs{\eta}\right)$, and $\mf{Y}^{**}_t=\mf{S}^{-1}\left(\mf{Z}^{**}_t\bs{\eta}_0+\alpha^{*}_t\mf{1}_n+\mf{U}^{*}_t\right)$. Thus, we have 
\begin{align*}
\mf{J}_n\mf{U}^{*}_t(\bs{\theta})=\mf{J}_n\left(\mf{S}(\bs{\rho})\mf{S}^{-1}\left(\mf{Z}^{**}_t\bs{\eta}_0+\alpha^{*}_t\mf{1}_n+\mf{U}^{*}_t\right)-\mf{Z}^{**}_t\bs{\eta}\right)
\end{align*}
Note that we can express $\mf{S}(\bs{\rho})\mf{S}^{-1}$ in the following way:
\begin{align*}
\mf{S}(\bs{\rho})\mf{S}^{-1}&=\left(\mf{I}_n-\sum_{l=1}^p\rho_{l}\mf{M}_l\right)\mf{S}^{-1}=\left(\mf{I}_n+\mf{S}-\mf{S}-\sum_{l=1}^p\rho_{l}\mf{M}_l\right)\mf{S}^{-1}\\
&=\left(\mf{S}-\sum_{l=1}^p(\rho_{l}-\rho_{l0})\mf{M}_l\right)\mf{S}^{-1}=\left(\mf{I}_n-\sum_{l=1}^p(\rho_{l}-\rho_{l0}) \mf{G}_{l}\right).
\end{align*}
Then, we can express $\mf{J}_n\mf{U}^{*}_t(\bs{\theta})$ as 
\begin{align*}
&\mf{J}_n\mf{U}^{*}_t(\bs{\theta})=\mf{J}_n\mf{S}(\bs{\rho})\mf{S}^{-1}\left(\mf{Z}^{**}_t\bs{\eta}_0+\alpha^{*}_t\mf{1}_n\right)+\mf{J}_n\mf{S}(\bs{\rho})\mf{S}^{-1}\mf{U}^{*}_t-\mf{J}_n\mf{Z}^{**}_t\bs{\eta}\\
&=\mf{J}_n\left(\mf{I}_n-\sum_{l=1}^p(\rho_{l}-\rho_{l0})\mf{G}_{l}\right)\left(\mf{Z}^{**}_t\bs{\eta}_0+\alpha^{*}_t\mf{1}_n\right)+\mf{J}_n\mf{S}(\bs{\rho})\mf{S}^{-1}\mf{U}^{*}_t-\mf{J}_n\mf{Z}^{**}_t\bs{\eta}\\
&=\mf{J}_n\left(\mf{Z}^{**}_t\left(\bs{\eta}_0-\bs{\eta}\right)-\sum_{l=1}^p(\rho_{l}-\rho_{l0})\mf{G}_{l}\left(\mf{Z}^{**}_t\bs{\eta}_0+\alpha^{*}_t\mf{1}_n\right)\right)+\mf{J}_n\mf{S}(\bs{\rho})\mf{S}^{-1}\mf{U}^{*}_t.
\end{align*}
Define $\mf{L}_{r,t}= \mf{G}_{r}\left(\mf{Z}^{**}_t\bs{\eta}_0+\alpha^{*}_t\mf{1}_n\right)$, $\mf{L}_{t}=\left(\mf{L}_{1,t},\hdots,\mf{L}_{p,t}\right)$, and $\mf{L}_{N}=\left(\mf{L}^{'}_{1},\hdots,\mf{L}^{'}_{T-1}\right)^{'}$. Thus, we can write $\mf{Q}^{'}_N\mf{J}_N\mf{U}_N(\bs{\theta})$ as
\begin{align*}
\mf{Q}^{'}_N\mf{J}_N\mf{U}_N(\bs{\theta})&=\mf{Q}^{'}_N\mf{J}_N\left(\mf{Z}_N\left(\bs{\eta}_0-\bs{\eta}\right)+\mf{L}_{N}\left(\bs{\rho}_0-\bs{\rho}\right)+\mf{S}_N(\bs{\rho})\mf{S}^{-1}_N\mf{U}^{*}_t\right)\\
&=\mf{Q}^{'}_N\mf{J}_N\left(\mf{Z}_N,\,\mf{L}_{N}\right)\left(\left(\bs{\eta}_0-\bs{\eta}\right)^{'},\left(\bs{\rho}_0-\bs{\rho}\right)^{'}\right)^{'}+\mf{Q}^{'}_N\mf{J}_N\mf{S}_N(\bs{\rho})\mf{S}^{-1}_N\mf{U}_N.
\end{align*}
By Lemma~\ref{l1}, we have $\plim_{n\to\infty}\frac{1}{N}\mf{Q}^{'}_N\mf{J}_N\mf{S}_N(\bs{\rho})\mf{S}^{-1}_N\mf{U}_N=\mf{0}$. Thus, we require that the equation $\plim_{n\to\infty}\frac{1}{N}\mf{Q}^{'}_N\mf{J}_N\left(\mf{Z}_N,\,\mf{L}_{N}\right)\left(\left(\bs{\eta}_0-\bs{\eta}\right)^{'},\left(\bs{\rho}_0-\bs{\rho}\right)^{'}\right)^{'}=\mf{0}$ should have a unique solution at $\bs{\theta}_0$. If $\plim_{n\to\infty}\frac{1}{N}\mf{Q}^{'}_N\mf{J}_N\left(\mf{Z}_N,\,\mf{L}_{N}\right)$ has the full column rank, then we will have a unique solution at $\bs{\theta}_0$. 

\section{Proof of Theorem~\ref{t1}}\label{pt1}
In this section, we will first show that  $\frac{1}{N}\frac{\partial \mf{g}_N(\bs{\theta}_0)}{\partial\bs{\theta}^{'}}=\mf{D}_{1N}+\mf{D}_{2N}+O_p(N^{-1/2})$. Recall that 
\begin{align}
\mf{g}_N(\bs{\theta})=
\begin{pmatrix}
\mf{U}^{'}_N(\bs{\theta})\mf{J}_N\mf{P}_{1N}\mf{J}_N\mf{U}_N(\bs{\theta})\\
\vdots\\
\mf{U}^{'}_N(\bs{\theta})\mf{J}_N\mf{P}_{mN}\mf{J}_N\mf{U}_N(\bs{\theta})\\
\mf{Q}^{'}_N\mf{J}_N\mf{U}_N(\bs{\theta})
\end{pmatrix},
\end{align}
where $\bs{\theta}=\left(\bs{\rho}^{'},\bs{\eta}^{'}\right)^{'}$, $\mf{U}_N(\bs{\theta})=\left(\mf{U}^{*'}_1(\bs{\theta}),\hdots,\mf{U}^{*'}_{T-1}(\bs{\theta})\right)^{'}$ and $\mf{U}^{*}_t(\bs{\theta})=\left(\mf{S}(\bs{\rho})\mf{Y}^{**}_t-\mf{Z}^{**}_t\bs{\eta}-\alpha^{*}_{t}\mf{1}_n\right)$. Note that
$\frac{\partial\mf{U}^{*}_t(\bs{\theta})}{\partial\rho_j}=-\mf{M}_j\mf{Y}^{**}_t$ for  $j=1,\hdots,p$ and $\frac{\partial\mf{U}^{*}_t(\bs{\theta})}{\partial\bs{\eta}^{'}}=-\mf{Z}^{**}_t$. Then, the components of $\frac{\partial \mf{g}_N(\bs{\theta})}{\partial\bs{\theta}^{'}}$ are
\begin{align}
\frac{\partial \mf{g}_N(\bs{\theta})}{\partial\bs{\eta}^{'}}=-
\begin{pmatrix}
\mf{U}^{'}_N(\bs{\theta})\mf{J}_N\mf{P}^s_{1N}\mf{J}_N\mf{Z}_N\\
\vdots\\
\mf{U}^{'}_N(\bs{\theta})\mf{J}_N\mf{P}^s_{mN}\mf{J}_N\mf{Z}_N\\
\mf{Q}^{'}_N\mf{J}_N\mf{Z}_N
\end{pmatrix},
\quad
\frac{\partial \mf{g}_N(\bs{\theta})}{\partial\bs{\rho}^{'}}=\left(\frac{\partial \mf{g}_N(\bs{\theta})}{\partial\rho_1},\hdots,\frac{\partial \mf{g}_N(\bs{\theta})}{\partial\rho_p}\right),
\end{align}
where
\begin{align}
\frac{\partial \mf{g}_N(\bs{\theta})}{\partial\rho_j}=
-
\begin{pmatrix}
\left(\mf{M}_{jN}\mf{Y}_N\right)^{'}\mf{J}_N\mf{P}^s_{1N}\mf{J}_N\mf{U}_N(\bs{\theta})\\
\vdots\\
\left(\mf{M}_{jN}\mf{Y}_N\right)^{'}\mf{J}_N\mf{P}^s_{mN}\mf{J}_N\mf{U}_N(\bs{\theta})\\
\mf{Q}^{'}_N\mf{J}_N\mf{M}_{jN}\mf{Y}_N
\end{pmatrix},
\end{align}
with $\mf{M}_{jN}=(\mf{I}_{T-1}\otimes\mf{M}_j)$ for $j=1,2,\hdots,p$. Next, we will determine the probability limit of $\frac{\partial \mf{g}_N(\bs{\theta}_0)}{\partial\bs{\theta}^{'}}$. Consider $\mf{U}^{'}_N\mf{J}_N\mf{P}^s_{lN}\mf{J}_N\mf{Z}_N$ in $\frac{\partial \mf{g}_N(\bs{\theta})}{\partial\bs{\eta}^{'}}$. Recall that $\mf{Z}^{**}_t=\left(\mf{Y}^{**,-1}_{t-1},\mb{M}\mf{Y}^{**,-1}_{t-1},\mf{X}^{*}_t\right)$. Then, we can express $\mf{Z}^{'}_N\mf{J}_N\mf{P}^s_{lN}\mf{J}_N\mf{U}_N$ in the following way
\begin{align}
\mf{Z}^{'}_N\mf{J}_N\mf{P}^s_{lN}\mf{J}_N\mf{U}_N&=\sum_{t=1}^{T-1}\mf{Z}^{**'}_t\mf{J}_n\mf{P}^s_{l}\mf{J}_n\mf{U}^{*}_t=
\begin{pmatrix}
\sum_{t=1}^{T-1}\mf{Y}^{**,-1'}_{t-1}\mf{J}_n\mf{P}^s_{l}\mf{J}_n\mf{U}^{*}_t\\
\sum_{t=1}^{T-1}\left(\mb{M}\Y^{**,-1}_{t-1}\right)^{'}\mf{J}_n\mf{P}^s_{l}\mf{J}_n\mf{U}^{*}_t\\
\sum_{t=1}^{T-1}\mf{X}^{*'}_t\mf{J}_n\mf{P}^s_{l}\mf{J}_n\mf{U}^{*}_t
\end{pmatrix}.
\end{align}
Using Lemma~\ref{l1}, it can be shown that
\begin{align*}
\E\left(\frac{1}{N}\sum_{t=1}^{T-1}\mf{\Y}^{**,-1}_{t-1}\mf{J}_n\mf{P}^s_{l}\mf{J}_n\mf{U}^{*}_t\right)=\frac{\sigma^2_0}{NT}\tr\left(\sum_{h=1}^{T-1}(T-h)\mf{A}^{h-1}\mf{S}^{-1}\mf{J}_n\mf{P}^s_l\mf{J}_n\right)
\end{align*}
\begin{align*}
\E\left(\frac{1}{N}\sum_{t=1}^{T-1}\left(\mf{M}_j\Y^{**,-1}_{t-1}\right)^{'}\mf{J}_n\mf{P}^s_{l}\mf{J}_n\mf{U}^{*}_t\right)=\frac{\sigma^2_0}{NT}\tr\left(\sum_{h=1}^{T-1}(T-h)\mf{A}^{h-1}\mf{S}^{-1}\mf{J}_n\mf{P}^s_l\mf{J}_n\mf{M}_j\right),
\end{align*}
\begin{align*}
\E\left(\frac{1}{N}\sum_{t=1}^{T-1}\mf{X}^{*'}_t\mf{J}_n\mf{P}^s_{l}\mf{J}_n\mf{U}^{*}_t\right)=\mf{0}_{k\times1}.
\end{align*}
for $l=1,2,\hdots,m$ and $j=1,2,\hdots,p$. Define $\mf{b}_{l\gamma}=\frac{\sigma^2_0}{N}\tr\left(\sum_{h=1}^{T-1}(T-h)\mf{A}^{h-1}\mf{S}^{-1}\mf{J}_n\mf{P}^s_l\mf{J}_n\right)$, $\mf{b}_{l\delta_j}=\frac{\sigma^2_0}{N}\tr\left(\sum_{h=1}^{T-1}(T-h)\mf{A}^{h-1}\mf{S}^{-1}\mf{J}_n\mf{P}^s_l\mf{J}_n\mf{M}_s\right)$ for $j=1,2,\hdots,p$, and 
\begin{align}
\mf{b}_{l\bs{\eta}}=\left(\mf{b}_{l\gamma},\mf{b}_{l\delta_1},\mf{b}_{l\delta_2},\hdots,\mf{b}_{l\delta_p},\mf{0}_{k\times1}\right).
\end{align}
Then, Lemma~\ref{l1} ensures that 
\begin{align}
\frac{1}{N}\mf{Z}^{'}_N\mf{J}_N\mf{P}^s_{lN}\mf{J}_N\mf{U}_N=\frac{1}{T}\mf{b}^{'}_{l\bs{\eta}}+O_P(N^{-1/2}),
\end{align}
for $l=1,2,\hdots,m$. Let $\mf{b}_{\bs{\delta}}=\left(\mf{b}^{'}_{1\bs{\eta}},\mf{b}^{'}_{l\bs{\eta}},\hdots,\mf{b}^{'}_{m\bs{\eta}}\right)^{'}$. Then, we have 
\begin{align}\label{d8}
\frac{1}{N}\frac{\partial \mf{g}_N(\bs{\theta})}{\partial\bs{\eta}^{'}}=
-\frac{1}{N}
\begin{pmatrix}
\mf{0}_{m\times k_z}\\
\mf{Q}^{'}_N\mf{J}_N\mf{Z}_N
\end{pmatrix}
+
\frac{1}{T}
\begin{pmatrix}
\mf{b}_{\bs{\delta}}\\
\mf{0}_{k_q\times k_z}
\end{pmatrix}
+O_p(N^{-1/2}).
\end{align}
Next, we consider $\left(\mf{M}_{rN}\mf{Y}_N\right)^{'}\mf{J}_N\mf{P}^s_{lN}\mf{J}_N\mf{U}_N$ in  $\frac{\partial \mf{g}_N(\bs{\theta}_0)}{\partial\bs{\rho}^{'}}$. Recall that $\mf{M}_r\mf{Y}^{**}_t= \mf{G}_{r}\left(\mf{Z}^{**}_t\bs{\eta}_0+\alpha^{*}_{t0}\mf{1}_n\right)+\mf{G}_{r}\mf{U}^{*}_t$ for $r=1,\hdots,p$. This expression implies that  $\mf{M}_{rN}\mf{Y}_N=\mf{L}_{rN}+\mf{G}_{rN}\mf{U}_N$, where $\mf{G}_{rN}=(\mf{I}_{T-1}\otimes\mf{G}_{r})$ and $\mf{L}_{rN}=\left(\mf{L}^{'}_{r1},\mf{L}^{'}_{r2},\hdots,\mf{L}^{'}_{r,T-1}\right)^{'}$ with $\mf{L}_{rt}=\mf{G}_{r}\left(\mf{Z}^{**}_t\bs{\eta}_0+\alpha^{*}_{t0}\mf{1}_n\right)$. Thus, we can express $\left(\mf{M}_{rN}\mf{Y}_N\right)^{'}\mf{J}_N\mf{P}^s_{lN}\mf{J}_N\mf{U}_N$ as
\begin{align*}
\left(\mf{M}_{rN}\mf{Y}_N\right)^{'}\mf{J}_N\mf{P}^s_{lN}\mf{J}_N\mf{U}_N=\mf{L}^{'}_{rN}\mf{J}_N\mf{P}^s_{lN}\mf{J}_N\mf{U}_N+\left(\mf{G}_{rN}\mf{U}_N\right)^{'}\mf{J}_N\mf{P}^s_{lN}\mf{J}_N\mf{U}_N
\end{align*}
By Lemma~\ref{l1}, we have $\frac{1}{N}\left(\mf{G}_{rN}\mf{U}_N\right)^{'}\mf{J}_N\mf{P}^s_{lN}\mf{J}_N\mf{U}_N=\frac{\sigma^2_0}{N}\tr\left(\mf{G}^{'}_{rN}\mf{J}_N\mf{P}^s_{lN}\mf{J}_N\right)+O_p(N^{-1/2})$ for $r=1,2,\hdots,p$ and $l=1,2,\hdots,m$. Note that we can express $\mf{L}^{'}_{rN}\mf{J}_N\mf{P}^s_{lN}\mf{J}_N\mf{U}_N$ as
\begin{align*}  
&\mf{L}^{'}_{rN}\mf{J}_N\mf{P}^s_{lN}\mf{J}_N\mf{U}_N=\sum_{t=1}^{T-1}\mf{L}^{'}_{rt}\mf{J}_n\mf{P}^s_l\mf{J}_n\mf{U}_t=\sum_{t=1}^{T-1}\left(\mf{G}_{r}\left(\mf{Z}^{**}_t\bs{\eta}_0+\alpha^{*}_{t0}\mf{1}_n\right)\right)^{'}\mf{J}_n\mf{P}^s_l\mf{J}_n\mf{U}_t\\
&=\sum_{t=1}^{T-1}\left(\mf{G}_{r}\left(\gamma_0+\sum_{s=1}^p\delta_{s0}\mf{M}_s\right)\Y^{**,-1}_{t-1}+\mf{G}_{r}\mf{X}^{*}_{t}\bs{\beta}_0+\alpha^{*}_{t0}\mf{G}_{r}\mf{1}_n\right)^{'}\mf{J}_n\mf{P}^s_l\mf{J}_n\mf{U}_t
\end{align*}
Then, using Lemma~\ref{l1}, we obtain 
\begin{align*}
&\frac{1}{N}\mf{L}^{'}_{rN}\mf{J}_N\mf{P}^s_{lN}\mf{J}_N\mf{U}_N\\
&=\frac{\sigma^2_0}{NT}\tr\left(\left(\mf{G}_{r}\left(\gamma_0+\sum_{j=1}^p\delta_{j0}\mf{M}_j\right)\right)\left(\sum_{h=1}^{T-1}(T-h)\mf{A}^{h-1}\mf{S}^{-1}\mf{J}_n\mf{P}^s_l\mf{J}_n\right)\right).
\end{align*}
Let $\mf{b}_{l\rho_r}=\frac{\sigma^2_0}{NT}\tr\left(\left(\mf{G}_{r}\left(\gamma_0+\sum_{j=1}^q\delta_{j0}\mf{M}_j\right)\right)\left(\sum_{h=1}^{T-1}(T-h)\mf{A}^{h-1}\mf{S}^{-1}\mf{J}_n\mf{P}^s_l\mf{J}_n\right)\right)$ for $r=1,2,\hdots,p$ and $l=1,2,\hdots,m$.  Define the following matrices:
\begin{align}
&\mf{C}_{N}=
\begin{pmatrix}
\tr\left(\mf{G}^{'}_{1N}\mf{J}_N\mf{P}^s_{1N}\mf{J}_N\right)&\hdots&\tr\left(\mf{G}^{'}_{pN}\mf{J}_N\mf{P}^s_{1N}\mf{J}_N\right)\\
\vdots&\ddots&\vdots\\
\tr\left(\mf{G}^{'}_{1N}\mf{J}_N\mf{P}^s_{mN}\mf{J}_N\right)&\hdots&\tr\left(\mf{G}^{'}_{pN}\mf{J}_N\mf{P}^s_{mN}\mf{J}_N\right)
\end{pmatrix},
\end{align}
\begin{align}
\mf{b}_{\bs{\rho}}=
\begin{pmatrix}
\mf{b}_{1\rho_1}&\mf{b}_{1\rho_2}&\hdots&\mf{b}_{1\rho_p}\\
\vdots&\vdots&\ddots&\vdots\\
\mf{b}_{m\rho_1}&\mf{b}_{m\rho_2}&\hdots&\mf{b}_{m\rho_p}\\
\mf{0}_{k_q\times1}&\mf{0}_{k_q\times1}&\hdots&\mf{0}_{k_q\times1}
\end{pmatrix}.
\end{align}
Let $\mf{L}_{N}=\left(\mf{L}_{1N},\hdots,\mf{L}_{pN}\right)$. The preceding analysis indicates that 
\begin{align}\label{d11}
\frac{\partial \mf{g}_N(\bs{\theta})}{\partial\bs{\eta}^{'}}=-\frac{1}{N}
\begin{pmatrix}
\sigma^2_0\mf{C}_{N}\\
\mf{Q}^{'}_N\mf{J}_N\mf{L}_{N}
\end{pmatrix}
+\frac{1}{T}
\begin{pmatrix}
\mf{b}_{\bs{\rho}}\\
\mf{0}_{k_q\times p}
\end{pmatrix}.
\end{align}
Substituting \eqref{d8} and \eqref{d11} into  $\frac{1}{N}\frac{\partial \mf{g}_N(\bs{\theta}_0)}{\partial\bs{\theta}^{'}}=\left(\frac{1}{N}\frac{\partial \mf{g}_N(\bs{\theta}_0)}{\partial\bs{\rho}^{'}},\,\frac{1}{N}\frac{\partial \mf{g}_N(\bs{\theta}_0)}{\partial\bs{\eta}^{'}}\right)$, we obtain
\begin{align}
\frac{1}{N}\frac{\partial \mf{g}_N(\bs{\theta}_0)}{\partial\bs{\theta}^{'}}=\mf{D}_{1N}+\mf{D}_{2N}+o_p(1),
\end{align}
 where
\begin{align}
\mf{D}_{1N}=-\frac{1}{N}
\begin{pmatrix}
\sigma^2_0\mf{C}_{N}&\mf{0}_{m\times k_z}\\
\mf{Q}^{'}_N\mf{J}_N\mf{L}_{N}&\mf{Q}^{'}_N\mf{J}_N\mf{Z}_{N}
\end{pmatrix}
=O(1),
\end{align}
\begin{align}
\mf{D}_{2N}=\frac{1}{T}
\begin{pmatrix}
\mf{b}_{\bs{\rho}}&\mf{b}_{\bs{\delta}}\\
\mf{0}_{k_q\times p}&\mf{0}_{k_q\times k_z}
\end{pmatrix}
=O(T^{-1}).
\end{align}
Under our set of assumptions, Lemma~\ref{l3} and Cramer-Wold device suggest that $\frac{1}{\sqrt{N}}\frac{\partial \mf{g}_N(\bs{\theta}_0)}{\partial\bs{\theta}^{'}}\xrightarrow{d}N\left(\mf{0},\,\plim_{n\to\infty}\bs{\Omega}_N\right)$. Then, the asymptotic distribution of the GMM estimator in our setting follows from the asymptotic argument in \citet[Proposition 2]{Lee:2007}.    
\section{Proof of Theorem~\ref{t2}}\label{pt2}
When $T$ is large, the precision matrix of $\sqrt{N}\left(\hat{\bs{\theta}}_N-\bs{\theta}_0\right)$ reduces to 
\begin{align}\label{E.1}
\mf{D}^{'}_{1N}\bs{\Omega}^{-1}_N\mf{D}_{1N}&=\frac{1}{N}
\begin{pmatrix}
\mf{C}^{'}_N\left(\bs{\Delta}_{mN}+\frac{\mu_4-3\sigma^4_0}{\sigma^4_0}\bs{\omega}^{'}_{mN}\bs{\omega}_{mN}\right)^{-1}\mf{C}_N&\mf{0}_{p\times k_z}\\
\mf{0}_{k_z\times p}&\mf{0}_{k_z\times k_z}
\end{pmatrix}\nonumber\\
&+\frac{1}{N\sigma^2_0}\left(\mf{L}_N,\, \mf{Z}_N\right)^{'}\mf{M}_{\mf{Q}}\left(\mf{L}_N,\, \mf{Z}_N\right).
\end{align}
As shown by \citet{Lee:2014}, the maximum of $\mf{C}^{'}_N\left(\bs{\Delta}_{mN}+\frac{\mu_4-3\sigma^4_0}{\sigma^4_0}\bs{\omega}^{'}_{mN}\bs{\omega}_{mN}\right)^{-1}\mf{C}_N$ is obtained by choosing $\mf{P}^{*}_j$ for $j=1,2,\hdots,p$. Thus, the maximum is 
\begin{align}\label{E.2}
\mf{C}^{*}_{N}=
\begin{pmatrix}
\tr\left(\mf{G}^{'}_{1N}\mf{J}_N\mf{P}^{*s}_{1N}\mf{J}_N\right)&\hdots&\tr\left(\mf{G}^{'}_{pN}\mf{J}_N\mf{P}^{*s}_{1N}\mf{J}_N\right)\\
\vdots&\ddots&\vdots\\
\tr\left(\mf{G}^{'}_{1N}\mf{J}_N\mf{P}^{*s}_{pN}\mf{J}_N\right)&\hdots&\tr\left(\mf{G}^{'}_{pN}\mf{J}_N\mf{P}^{*s}_{pN}\mf{J}_N\right)
\end{pmatrix}.
\end{align}
Next , we consider the second term on the right hand side of \eqref{E.1}. It follows from Lemma~\ref{l1} (4) that $\plim_{n,T\to\infty}\frac{1}{N}\mf{Q}^{'}_N\mf{J}_N\left(\mf{L}_N,\mf{Z}^{**}_N\right)=\plim_{n,T\to\infty}\frac{1}{N}\mf{Q}^{'}_N\mf{J}_N\mf{Q}^{*}_N$, where $\mf{Q}^{*}_N=(\mf{Q}^{*'}_1,\hdots,\mf{Q}^{*'}_{T-1})^{'}$. Thus, we have 
$$
\plim_{n,T\to\infty}\frac{1}{N}\left(\mf{L}_N,\, \mf{Z}_N\right)^{'}\mf{M}_{\mf{Q}}\left(\mf{L}_N,\, \mf{Z}_N\right)\leq \plim_{n,T\to\infty}\frac{1}{N}\mf{Q}^{*'}_N\mf{J}_N\mf{Q}^{*}_N,
$$
suggesting that $\mf{Q}^{*}_N$ is the best IV matrix.  From Lemma 5 of \citet{Lee:2014}, it also follows that 
\begin{align}\label{E.3}
\plim_{n,T\to\infty}\frac{1}{N}\left(\mf{L}_N,\, \mf{Z}_N\right)^{'}\mf{J}_{N}\left(\mf{L}_N,\, \mf{Z}_N\right)=\plim_{n,T\to\infty}\frac{1}{N}\mf{Q}^{*'}_N\mf{J}_N\mf{Q}^{*}_N.
\end{align}
Thus, the results in \eqref{E.2} and \eqref{E.3} suggest that that the precision matrix of $\sqrt{N}\left(\hat{\bs{\theta}}^{*}_N-\bs{\theta}_0\right)$ is given by $\bs{\Sigma}^{*}_N$. Finally, the asymptotic distribution of $\sqrt{N}\left(\hat{\bs{\theta}}^{*}_N-\bs{\theta}_0\right)$ can be shown by following the argument given in Theorem 2 of \citet{Lee:2014}. Therefore, the details are omitted.

\newpage

\bibliographystyle{agsm}

\begin{thebibliography}{}

\bibitem[Abadir and Magnus, 2005]{Karim:2005}
Abadir, K.~M. and Magnus, J.~R. (2005).
\newblock {\em Matrix Algebra}.
\newblock Cambridge University Press, New York.

\bibitem[Bashar, 2021]{bashar2021intra}
Bashar, O.~H. (2021).
\newblock An intra-city analysis of house price convergence and spatial
  dependence.
\newblock {\em The Journal of Real Estate Finance and Economics},
  63(4):525--546.

\bibitem[Bill{\'e} et~al., 2017]{bille2017two}
Bill{\'e}, A.~G., Benedetti, R., and Postiglione, P. (2017).
\newblock A two-step approach to account for unobserved spatial heterogeneity.
\newblock {\em Spatial Economic Analysis}, 12(4):452--471.

\bibitem[Bollerslev et~al., 1992]{Bollerslev:1992}
Bollerslev, T., Chou, R.~Y., and Kroner, K.~F. (1992).
\newblock Arch modeling in finance: A review of the theory and empirical
  evidence.
\newblock {\em Journal of Econometrics}, 52(1):5--59.

\bibitem[Chang and Diao, 2021]{chang2021inter}
Chang, Z. and Diao, M. (2021).
\newblock Inter-city transport infrastructure and intra-city housing markets:
  {Estimating the redistribution effect of high-speed rail in Shenzhen, China}.
\newblock {\em Urban Studies}.

\bibitem[Engle, 1982]{Engle:1982}
Engle, R.~F. (1982).
\newblock Autoregressive conditional heteroscedasticity with estimates of the
  variance of {United Kingdom} inflation.
\newblock {\em Econometrica}, 50(4):987--1007.

\bibitem[Engle and Bollerslev, 1986]{Engle:1986}
Engle, R.~F. and Bollerslev, T. (1986).
\newblock Modelling the persistence of conditional variances.
\newblock {\em Econometric Reviews}, 5(1):1--50.

\bibitem[Gupta and Robinson, 2015]{gupta2015inference}
Gupta, A. and Robinson, P.~M. (2015).
\newblock Inference on higher-order spatial autoregressive models with
  increasingly many parameters.
\newblock {\em Journal of Econometrics}, 186(1):19--31.

\bibitem[H{\o}lleland and Karlsen, 2020]{Hol:2020}
H{\o}lleland, S. and Karlsen, H.~A. (2020).
\newblock A stationary spatio-temporal {GARCH} model.
\newblock {\em Journal of Time Series Analysis}, 41(2):177--209.

\bibitem[Holmes et~al., 2017]{holmes2017pair}
Holmes, M.~J., Otero, J., and Panagiotidis, T. (2017).
\newblock A pair-wise analysis of intra-city price convergence within the
  {Paris} housing market.
\newblock {\em The Journal of Real Estate Finance and Economics}, 54(1):1--16.

\bibitem[Jacquier et~al., 1994]{JPR:1994}
Jacquier, E., Polson, N.~G., and Rossi, P.~E. (1994).
\newblock Bayesian analysis of stochastic volatility models.
\newblock {\em Journal of Business \& Economic Statistics}, 12(4):371--389.

\bibitem[{Kelejian} and {Prucha}, 2010]{KP:2010}
{Kelejian}, H.~H. and {Prucha}, I.~R. (2010).
\newblock Specification and estimation of spatial autoregressive models with
  autoregressive and heteroskedastic disturbances.
\newblock {\em Journal of Econometrics}, 157:53--67.

\bibitem[Kim et~al., 1998]{Kim:1998}
Kim, S., Shephard, N., and Chib, S. (1998).
\newblock Stochastic volatility: Likelihood inference and comparison with
  {ARCH} models.
\newblock {\em The Review of Economic Studies}, 65(3):361--393.

\bibitem[Lee, 2004]{Lee:2004}
Lee, L.-f. (2004).
\newblock Asymptotic distributions of quasi-maximum likelihood estimators for
  spatial autoregressive models.
\newblock {\em Econometrica}, 72(6):1899--1925.

\bibitem[Lee, 2007]{Lee:2007}
Lee, L.-f. (2007).
\newblock {GMM} and {2SLS} estimation of mixed regressive, spatial
  autoregressive models.
\newblock {\em Journal of Econometrics}, 137(2):489--514.

\bibitem[Lee and Liu, 2010]{lee2010efficient}
Lee, L.-f. and Liu, X. (2010).
\newblock Efficient {GMM} estimation of high order spatial autoregressive
  models with autoregressive disturbances.
\newblock {\em Econometric Theory}, 26(1):187--230.

\bibitem[Lee and Yu, 2010]{Yu:2010}
Lee, L.-f. and Yu, J. (2010).
\newblock A spatial dynamic panel data model with both time and individual
  fixed effects.
\newblock {\em Econometric Theory}, 26:564--597.

\bibitem[Lee and Yu, 2014]{Lee:2014}
Lee, L.-f.~L. and Yu, J. (2014).
\newblock Efficient {GMM} estimation of spatial dynamic panel data models with
  fixed effects.
\newblock {\em Journal of Econometrics}, 180(2):174--197.

\bibitem[LeSage and Pace, 2009]{Lesage:2009}
LeSage, J.~P. and Pace, R.~K. (2009).
\newblock {\em Introduction to Spatial Econometrics (Statistics: A Series of
  Textbooks and Monographs}.
\newblock Chapman and Hall/CRC, London.

\bibitem[McMillen, 2014]{mcmillen2014local}
McMillen, D. (2014).
\newblock Local quantile house price indices.
\newblock {\em Journal of Urban Economics}.

\bibitem[Meen, 1999]{Meen:1999}
Meen, G. (1999).
\newblock Regional house prices and the ripple effect: A new interpretation.
\newblock {\em Housing Studies}, 14(6):733--753.

\bibitem[Merk and Otto, 2021]{merk2021directional}
Merk, M.~S. and Otto, P. (2021).
\newblock Directional spatial autoregressive dependence in the conditional
  first-and second-order moments.
\newblock {\em Spatial Statistics}, 41:100490.

\bibitem[Otto, 2019]{Otto:2019}
Otto, P. (2019).
\newblock {spGARCH: An R-Package for Spatial and Spatiotemporal ARCH and GARCH
  models}.
\newblock {\em {The R Journal}}, 11(2):401--420.

\bibitem[Otto and Schmid, 2020]{Otto:2020}
Otto, P. and Schmid, W. (2020).
\newblock Spatial and spatiotemporal {GARCH} models -- a unified approach.

\bibitem[Otto et~al., 2018]{Otto:2018}
Otto, P., Schmid, W., and Garthoff, R. (2018).
\newblock Generalised spatial and spatiotemporal autoregressive conditional
  heteroscedasticity.
\newblock {\em Spatial Statistics}, 26:125--145.

\bibitem[Robinson, 2009]{Robinson:2009}
Robinson, P.~M. (2009).
\newblock Large-sample inference on spatial dependence.
\newblock {\em Econometrics Journal}, 12.

\bibitem[Sandmann and Koopman, 1998]{Koopman:1998}
Sandmann, G. and Koopman, S.~J. (1998).
\newblock Estimation of stochastic volatility models via {Monte Carlo} maximum
  likelihood.
\newblock {\em Journal of Econometrics}, 87(2):271 -- 301.

\bibitem[Sato and Matsuda, 2017]{Sato:2017}
Sato, T. and Matsuda, Y. (2017).
\newblock Spatial autoregressive conditional heteroskedasticity models.
\newblock {\em Journal of the Japan Statistical Society}, 47(2):221--236.

\bibitem[Sato and Matsuda, 2021]{Takaki:2021}
Sato, T. and Matsuda, Y. (2021).
\newblock Spatial extension of generalized autoregressive conditional
  heteroskedasticity models.
\newblock {\em Spatial Economic Analysis}, 16(2):148--160.

\bibitem[Shephard, 1994]{Shephard:1994}
Shephard, N. (1994).
\newblock Partial non-{Gaussian} state space.
\newblock {\em Biometrika}, 81(1):115--131.

\bibitem[Ta{\c s}p{\i}nar et~al., 2021]{Taspinar:2021}
Ta{\c s}p{\i}nar, S., Do{\u g}an, O., Chae, J., and Bera, A.~K. (2021).
\newblock Bayesian inference in spatial stochastic volatility models: An
  application to house price returns in {Chicago}.
\newblock {\em Oxford Bulletin of Economics and Statistics}, 83:1243--1272.

\bibitem[Tobler, 1970]{Tobler70}
Tobler, W.~R. (1970).
\newblock A computer movie simulating urban growth in the {D}etroit region.
\newblock {\em Economic Geography}, 46:234--240.

\bibitem[Yu et~al., 2008]{Yu:2008}
Yu, J., de~Jong, R., and fei Lee, L. (2008).
\newblock Quasi-maximum likelihood estimators for spatial dynamic panel data
  with fixed effects when both n and t are large.
\newblock {\em Journal of Econometrics}, 146(1):118 -- 134.

\bibitem[Zhang and Yi, 2017]{zhang2017quantile}
Zhang, L. and Yi, Y. (2017).
\newblock Quantile house price indices in {Beijing}.
\newblock {\em Regional Science and Urban Economics}, 63:85--96.

\end{thebibliography}

\end{document}